\pdfoutput=1
\documentclass[11pt]{article}

\usepackage{geometry}
\PassOptionsToPackage{numbers, compress}{natbib}
\usepackage{comment}

\usepackage[utf8]{inputenc} 
\usepackage[T1]{fontenc}    
\usepackage{url}            
\usepackage{booktabs}       
\usepackage{amsfonts}       
\usepackage{amsthm}
\usepackage{nicefrac}       
\usepackage{microtype}      
\usepackage{xcolor}         
\usepackage{amsmath}
\usepackage{float}
\usepackage{enumitem}

\usepackage[linktocpage=true,
pagebackref=true,colorlinks,citecolor=blue,
bookmarks,bookmarksopen,bookmarksnumbered]
{hyperref}
\usepackage[noabbrev, nameinlink]{cleveref}
\crefname{property}{property}{Property}
\creflabelformat{property}{(#1)#2#3}
\crefname{equation}{eq}{Eq}
\creflabelformat{equation}{(#1)#2#3}
\crefname{algocf}{Algorithm}{Algorithm}
\creflabelformat{algocf}{#1#2#3}

\usepackage{multirow}
\usepackage{placeins}

\usepackage{amsthm}
\usepackage{amssymb,amsmath}
\usepackage{bbm}
\usepackage{xcolor}
\usepackage{microtype}
\usepackage{graphicx}
\usepackage[ruled,vlined]{algorithm2e}
\usepackage{mdframed}
\usepackage{caption,subcaption}
\usepackage{verbatim}

\usepackage{tcolorbox}
\tcbuselibrary{skins,breakable}
\tcbset{enhanced jigsaw}

\definecolor{ceruleanblue}{rgb}{0.16, 0.32, 0.75}
\definecolor{darkmidnightblue}{rgb}{0.0, 0.2, 0.4}
\definecolor{darkpastelgreen}{rgb}{0.01, 0.75, 0.24}
\definecolor{bleudefrance}{rgb}{0.19, 0.55, 0.91}
\definecolor{RURed}{rgb}{0.65, 0.2, 0.15}
\newcommand{\chen}[1]{\textcolor{RURed}{[\textbf{Chen:} #1]}}

\newcommand{\ex}[1]{\mathbb{E}\left[#1\right]}
\newcommand{\tw}{\tilde{w}}
\newcommand{\td}{\tilde{d}}
\newcommand{\tT}{\tilde{T}}
\newcommand{\cor}{\textsc{Corrupt}}

\newcommand{\OPT}{\ensuremath{\textsf{OPT}}}

\newcommand{\card}[1]{\ensuremath{|#1|}}
\newcommand{\eps}{\ensuremath{\varepsilon}}
\newcommand{\dtilde}{\ensuremath{\widetilde{d}}}
\newcommand{\dest}{\ensuremath{d^{\text{est}}}}
\newcommand{\med}[1]{\ensuremath{\textsf{Median}}\{#1\}}

\DeclareMathOperator*{\R}{\mathbb{R}}

\DeclareMathOperator*{\argmax}{arg\,max}

\makeatletter
\newtheorem*{rep@theorem}{\rep@title}
\newcommand{\newreptheorem}[2]{%
	\newenvironment{rep#1}[1]{%
		\def\rep@title{\Cref{##1}}%
		\begin{rep@theorem}}%
		{\end{rep@theorem}}}
\makeatother

\newtheorem{theorem}{Theorem}
\newreptheorem{theorem}{Theorem}
\newtheorem{lemma}[theorem]{Lemma}
\newreptheorem{lemma}{Lemma}
\newtheorem{remark}[theorem]{Remark}
\newtheorem{definition}{Definition}
\newtheorem{claim}[theorem]{Claim}

\newtheorem{fact}[theorem]{Fact}
\newtheorem{observation}[theorem]{Observation}

\newtheorem{proposition}[theorem]{Proposition}

\RestyleAlgo{ruled}
\SetKwInput{KwData}{Input}
\SetKwInput{KwResult}{Output}

\usepackage{nameref}

\newcommand{\bracket}[1]{\left[#1\right]}
\newcommand{\paren}[1]{\ensuremath{\left(#1\right)}\xspace}

\newcommand{\IR}{\ensuremath{\mathbb{R}}}

\newcommand{\prob}[1]{\Pr\paren{#1}}
\newcommand{\expect}[1]{\Exp\bracket{#1}}

\newcommand{\poly}{\mbox{\rm poly}}
\newcommand{\polylog}{\mbox{\rm  polylog}\,}

\newcommand{\ALG}{\ensuremath{\mbox{\sc alg}}\xspace}

\DeclareMathOperator*{\Exp}{\ensuremath{{\mathbb{E}}}}
\DeclareMathOperator*{\Prob}{\ensuremath{\textnormal{Pr}}}
\renewcommand{\Pr}{\Prob}

\newenvironment{tbox}{\begin{tcolorbox}[
		enlarge top by=5pt,
		enlarge bottom by=5pt,
		 breakable,
		 boxsep=0pt,
                  left=4pt,
                  right=4pt,
                  top=10pt,
                  arc=0pt,
                  boxrule=1pt,toprule=1pt,
                  colback=white
                  ]
	}
{\end{tcolorbox}}

\newcommand{\istar}{\ensuremath{i^{\star}}}

\newcommand{\cstar}{\ensuremath{c^{*}}\xspace}
\newcommand{\rstar}{\ensuremath{r^{*}}\xspace}

\newcommand{\ctilde}{\ensuremath{\tilde{c}}\xspace}

\newcommand{\DQ}{\ensuremath{\textnormal{\textsf{WO}}}\xspace}
\newcommand{\LQ}{\ensuremath{\textnormal{\textsf{SO}}}\xspace}
\newcommand{\SO}{\ensuremath{\textnormal{\textsf{SO}}}\xspace}
\newcommand{\WO}{\ensuremath{\textnormal{\textsf{WO}}}\xspace}

\newcommand{\Tstar}{\ensuremath{T^{*}}\xspace}

\newcommand{\ball}[2]{\ensuremath{B(#1, #2)}\xspace}
\newcommand{\lball}[3]{\ensuremath{B^{#1}(#2, #3)}\xspace}
\newcommand{\ldball}[4]{\ensuremath{B_{#1}^{#2}(#3, #4)}\xspace}

\newcommand{\OPTcen}{\ensuremath{\OPT_{\text{k-center}}}\xspace}
\newcommand{\OPTmean}{\ensuremath{\OPT_{\text{k-means}}}\xspace}
\newcommand{\OPTmed}{\ensuremath{\OPT_{\text{k-median}}}\xspace}

\newcommand{\cost}{\ensuremath{\textnormal{\textsf{cost}}}\xspace}

\newcommand{\myqed}[1]{\let\qed\relax \hspace*{\fill} #1 \ensuremath{\square}}

\newcommand{\cA}{\mathcal{A}}
\newcommand{\cB}{\mathcal{B}}
\newcommand{\cC}{\mathcal{C}}

\newcommand{\cE}{\mathcal{E}}

\newcommand{\cH}{\mathcal{H}}

\newcommand{\cQ}{\mathcal{Q}}

\newcommand{\cX}{\mathcal{X}}

\newtheorem{innercustomthm}{Theorem}
\newenvironment{customthm}[1]
  {\renewcommand\theinnercustomthm{#1}\innercustomthm}
  {\endinnercustomthm}

\title{
Metric Clustering and MST with Strong and Weak Distance Oracles
}

\author{%
  MohammadHossein Bateni \\
  Google Research \\
  \texttt{bateni@google.com} \\
  \and
  Prathamesh Dharangutte \\
  Rutgers University \\
  \texttt{prathamesh.d@rutgers.edu} \\
  \and
  Rajesh Jayaram \\
  Google Research \\
  \texttt{rkjayaram@google.com} \\
  \and
  Chen Wang \\
  Rutgers University \\
  \texttt{chen.wang.cs@rutgers.edu} \\
}

\DontPrintSemicolon

\begin{document}

\maketitle
\begin{abstract}

We study optimization problems in a metric space $(\mathcal{X},d)$ where we can compute distances in two ways: via a ``strong'' oracle that returns exact distances $d(x,y)$, and a ``weak'' oracle that returns distances $\tilde{d}(x,y)$ which may be arbitrarily corrupted with some probability. This model captures the increasingly common trade-off between employing both an expensive similarity model (e.g. a large-scale embedding model),
and a less accurate but cheaper model. Hence, the goal is to make as few queries to the strong oracle as possible. We consider both so-called ``point queries'', where the strong oracle is queried on a set of points $S \subset \cX $ and returns $d(x,y)$ for all $x,y \in S$, and ``edge queries'' where it is queried for individual distances $d(x,y)$. 


    Our main contributions are optimal algorithms and lower bounds for clustering and Minimum Spanning Tree (MST) in this model. 
    For $k$-centers, $k$-median, and $k$-means, we give constant factor approximation algorithms with only $\tilde{O}(k)$ strong oracle point queries, and prove that $\Omega(k)$ queries are required for any bounded approximation. For edge queries, our upper and lower bounds are both $\tilde{\Theta}(k^2)$.  Surprisingly, for the MST problem we give a $O(\sqrt{\log n})$ approximation algorithm using no strong oracle queries at all, and a matching $\Omega(\sqrt{\log n})$ lower bound. 
 We empirically evaluate our algorithms, and show that their quality is comparable to that of the baseline algorithms that are given all true distances, but while querying the strong oracle on only a small fraction ($<1\%$) of points. 



\end{abstract}

 \renewenvironment{quote}{%
   \list{}{%
     \leftmargin0.1cm   
     \rightmargin\leftmargin
   }
   \item\relax
 }
 {\endlist}

\section{Introduction}
\label{sec:intro}
Large-scale similarity models are ubiquitous in modern machine learning, where they are used to generate real-valued distances for non-metric data, such as images, text, and videos. A popular example is embedding models \cite{mikolov2013efficient,van2008visualizing,he2016deep,devlin2018bert},  which transform a data point $x$ into a point $f(x)$ in a metric space $(\cX,d)$, such that the similarity between $x,y$ can be inferred by the distance $d(f(x),f(y))$. 
    However, as the scale and quality of these models grow, so too increases the
    resources required to run them. Thus, a common component of many ML pipelines is to additionally employ an efficient but less precise similarity model to reduce the number of expensive distance comparisons made with the more accurate model \cite{kusner2015word,9191120}. Common examples of such ``weak'' secondary similarity models include hand-crafted models based on simple features (location, timestamp, bitrate, etc.), lightweight neural network, models trained on cheap but sometimes inaccurate data \cite{9191120}, meta-data obtained in video transcoding \cite{9191120,ringis2021near}, previously computed similarities from historical data \cite{mitzenmacher2022algorithms}, and the retrieve-then-rerank architecture for recommendation systems \cite{LiuXQS00ZT22}, text retrieval \cite{ZLXX22RR}, question-answering \cite{BarzS19} and vision-applications \cite{ZhongZCL17}.

 
Understanding the complexity of computational tasks in the presence of noisy or imprecise oracles is a fundamental problem dating back multiple decades \cite{feige1994computing}, and many problems such as clustering, sorting, and nearest neighbor search have been intensively studied therein \cite{braverman2008noisy, braverman2009sorting,mazumdar2017clustering,green2020clustering, mason2019learning}. However, despite the popularity of combining two oracles in practice, the majority of this work considers only a single imprecise oracle, whereas less work has been done to understand the complexity of tasks using \textit{both} a noisy (weak) oracle, and an exact (strong) oracle. In this paper, we initiate a formal study of this setting for metric optimization problems. 

Specifically, we introduce the \emph{Weak-Strong Oracle Model}: here, we are given a metric space $(\cX, d)$ of $|\cX| = n$ points, where $d: \, \cX \times \cX \rightarrow \IR$ is the underlying metric, representing the output of an expensive but accurate similarity model,\footnote{We often use ``similarity'' and ``distance'' interchangeably, as similarity models, especially embedding-based models, can usually be easily converted to distance models and vice versa. } 
as well as a corruption probability $\delta \in (0,1/2)$. The metric $d$ is not known to the algorithm a priori, but can be accessed through two types of queries: \textit{strong} and \textit{weak oracle} queries. For the strong oracle, we consider two possible forms of queries: edge queries and point queries. These queries are defined as follows:

\begin{itemize}[leftmargin=*]
\item \textbf{Weak oracle queries ($\WO(x,y)$):} given $(x,y) \in \cX^2$, the weak oracle returns a value $\dtilde(x,y)$ such that: with probability $1-\delta$ we have $\td(x,y) = d(x,y)$, and otherwise, with probability $\delta$, the value $\td(x,y)$ is set arbitrarily.\footnote{A pair $(x,y)$ such that the weak oracle distance $\td(x,y)$ can be set arbitrarily is called \textit{corrupted}. } The randomness is \emph{independent} across different pairs $(x,y)$, and drawn \emph{exactly once} (i.e., repeated queries to $\tilde{d}(x,y)$ will yield the same result).
\item \textbf{Strong oracle (point) queries ($\SO(x)$):} given a point $x \in \cX$, a strong point oracle returns a symbolic value $\SO(x)$. The value $\SO(x)$ gives no information on its own. However, given any two values $\SO(x),\SO(y)$, the algorithm can compute the true distance $d(x,y)$.

\item \textbf{Strong oracle (edge) queries ($\SO(x,y)$):} given $x,y \in \cX$, a strong edge oracle returns the true distance $\SO(x,y) = d(x,y)$.
\end{itemize}


The weak oracle distances $\td$ capture a cheap but less precise distance model, whereas the strong oracle is considered to be significantly more expensive. As a result, our goal is to produce a high-quality solution to an optimization problem (e.g. clustering) for the underlying metric $(\cX,d)$ while \textit{minimizing} the number of queries made to the strong oracle. We even allow the corruptions that occur to the weak oracle to be \textit{adversarial} (see Section \ref{sec:prelim} for precise model definitions); this captures a very general class of ``imprecise weak oracles'', allowing them to produce arbitrary bad distances with some probability. 

Depending on the context, it may make sense to allow only one of the two types of strong oracle queries. Therefore, we consider two models: \textbf{(1)} where only strong oracle point queries are allowed, and \textbf{(2)} where only strong oracle edge queries are allowed. Note that the two types of strong oracle queries are closely related. In particular, any algorithm that makes $q$ strong oracle point queries can be simulated by an algorithm that makes $q^2$ strong oracle edge queries. In this paper, we will give algorithms and lower bounds for both strong oracle query models. We recall that an important motivation for strong oracle point queries is the example of an expensive embedding model: in this case, $\SO(x)$ represents the embedding into $(\cX,d)$. Conversely, the strong oracle edge query model is natural in settings where the expensive model can compute pair-wise similarities, such as cross-attention models \cite{brown2020language,thoppilan2022lamda}.  

 We focus on clustering, which is one of the most fundamental unsupervised learning tasks, and the classic metric minimum spanning tree (MST) problem, which has applications to network design and hierarchical clustering. Both tasks have been studied extensively in the literature on noisy oracles \cite{ashtiani2016clustering, mazumdar2017clustering, ailon2017approximate, ergun2022learningaugmented, nguyen2023improved, silwal2023kwikbucks}. However, given the strong type of inaccuracies allowed by our weak oracle, a priori it is not clear whether we can solve these foundational tasks without querying the strong oracle for essentially all the distances. Specifically, we pose the following question:

\begin{quote}
 \begin{center}
  {\it  Is it possible to solve metric optimization tasks, like clustering and MST, in the Weak-Strong Oracle Model while making fewer that $\Omega(n)$ strong oracle point queries (or $\Omega(n^2)$ edge queries)?}
 \end{center}
 \end{quote}

\subsection{Contributions}\label{sec:contributions}
Our main contribution is to answer the above question in the affirmative. Specifically, we design constant factor approximation algorithms for $k$-centers, $k$-means, $k$-medians with $\tilde{O}(k)$\footnote{Throughout, we write $\tilde{O}$ to suppress $\log n$ factors.} point queries to the strong oracle. For MST, we design an algorithm that achieves a $O(\sqrt{\log n})$ approximation without \textit{any} strong oracle queries. For both problems, we prove matching or nearly matching lower bounds, demonstrating the optimality of our algorithms. 
Our results for $k$-clustering hold for any corruption probability $\delta \in (0,1/2)$ bounded away from $1/2$ by a constant, and for MST our results hold for any $\delta \in (0,1)$ bounded away from $1$ by a constant. 

\paragraph{Clustering.}
We begin with our results for $k$-clustering. Here, our goal is to produce a set of $k$ \textit{centers} $c_1,\dots,c_k \in \cX$, as well as a mapping $\cC:\cX \to \{c_i\}_{i=1}^k$, so as to minimize the $k$-clustering cost \textit{with respect to the original metric $(\cX,d)$}. Recall that for $k$-centers, the objective is to minimize $\max_{p \in \cX} d(p,\cC(p))$. For the other two objectives, the goal is to minimize $\sum_{p \in \cX} d^q(p, \cC(p))$, where $q=1$ for $k$-median and $q=2$ for $k$-means. Our results for $k$-clustering tasks are as follows:



\begin{customthm}{\ref{thm:alg-k-center} and \ref{thm:alg-k-means}}[Clustering Upper Bounds]
There exists algorithms in the weak-strong oracle model that, with high probability, obtain $O(1)$ approximations to $k$-centers, $k$-means, and $k$-median. The algorithms use $\tilde{O}(k)$ strong oracle point queries, or $\tilde{O}(k^2)$ edge queries, and run in time $\tilde{O}(nk)$. 
\end{customthm}



Despite the similar query-complexities, the clustering algorithms from Theorems \ref{thm:alg-k-center} and \ref{thm:alg-k-means} require very different techniques. 
Moreover, since it is NP-Hard to give better than a $2$ approximation to any of the above clustering tasks \cite{hsu1979easy,cohen2019inapproximability,cohen2022johnson}, our algorithm's approximations are optimal up to a constant. Next, we show that the query complexity of our algorithms are nearly optimal, even for arbitrarily large approximations, which settles the complexity of this problem up to $\log n$-factors.



\begin{customthm}{\ref{thm:k-clustering-lb}}[Clustering Lower Bound]
Any algorithm which obtains a multiplicative $c$-approximation, for any approximation factor $c$, with probability at least $1/2$, to either $k$-centers, $k$-means, or $k$-medians, must make at least $\Omega(k^2)$ strong oracle edge queries, or $\Omega(k)$ strong oracle point queries.
\end{customthm}    

\paragraph{Minimum Spanning Tree.} In the classic metric MST problem, the goal is to produce a spanning tree $T$ of the points in $\cX$ so as to minimize the weight of the tree in the original metric $(\cX,d)$: namely $w(T) = \sum_{\text{edge } (x,y) \in T} d(x,y)$. We consider the problem in two settings, corresponding to whether or not the weak oracle distances $\td:\cX \times \cX \to \R$ are themselves a metric over $\cX$. We refer to the case where $(\cX,\td)$ is restricted to being a metric as the \textit{metric-weak oracle} setting. This setting is especially motivated by weak oracles which are themselves embedding models, such as lighter-weight embeddings or pre-computed embeddings trained on stale or possibly inaccurate data. We demonstrate that, perhaps surprisingly, given a metric weak oracle, we can obtain a good approximation to the optimal MST without resorting to the strong oracle at all.
\begin{customthm}{\ref{thm:mst-alg-metric-case}}
\label{thm:mst-alg-metric-case}
There is an algorithm that, given only access to the distances $\td$ produced by a metric weak oracle (namely, $(\cX,\td)$ is metric), produces a tree $\hat{T}$ such that $\mathbb{E}[w(\hat{T})] \leq O(\sqrt{ \log n}) \cdot \min_T w(T)$.
\end{customthm}  

A natural question that arises following our MST algorithm is whether a constant approximation is possible, perhaps by allowing for a small number of strong oracle queries as well. We demonstrate, however, that this is impossible in a strong sense: any algorithm that achieves a better than $O(\sqrt{ \log n})$ approximation must essentially query the strong oracle for all the distances in $\cX$. 
\begin{customthm}{\ref{thm:mstLBMain}}
  There exists a constant $c$ such that any algorithm that outputs a spanning tree $\hat{T}$ such that $\mathbb{E}[w(\hat{T})]\leq c  \sqrt{\log n}\cdot \min_{T}w(T)$ must make at least $\Omega(n/\sqrt{\log n})$ queries to the strong oracle, and this holds even when $\td:\cX^2 \to \R$ is restricted to being a metric. 
\end{customthm}  
Thus, Theorems \ref{thm:mst-alg-metric-case} and \ref{thm:mstLBMain} prove tight bounds for the approximation of MST in the metric weak oracle setting. A final question is whether a $O(\sqrt{\log n})$ approximation is possible without the metric restriction on $\td$. We demonstrate that this too is impossible, by proving a $\Omega(\log n)$ lower bound on the approximation in the general case. Taken with our upper bound in Theorem \ref{thm:mst-alg-metric-case}, this proves a strong separation between the metric and non-metric weak oracle models.

\begin{customthm}{\ref{thm:mstLBGeneralMain}}
  There exists a constant $c$ such that any algorithm which outputs a spanning tree $\hat{T}$  such that $\mathbb{E}[w(\hat{T})]\leq c  \log n\cdot \min_{T}w(T)$ must make at least $\Omega(n)$ point queries to the strong oracle. 
\end{customthm}  

We conjecture that the lower bound from Theorem \ref{thm:mstLBGeneralMain} is tight, and that an algorithm exists with \textit{no} strong oracle queries and a $O(\log n)$ approximation. We leave this as an open question for future work to determine the exact query complexity of non-metric MST in the weak-strong oracle model. 


\paragraph{Experiments.} We empirically evaluate the performance of our algorithms on both synthetic and real-world datasets. 
For the synthetic data experiments, we use the extensively studied Stochastic Block Model \cite{holland1983stochastic,dyer1989solution,decelle2011asymptotic,abbe2015exact,abbe2015community,hajek2016achieving,mossel2015consistency}, which has a natural interpretation as clustering with faulty oracles \cite{mazumdar2017clustering}. 
For the real-world dataset, we run experiments to cluster embeddings of the MNIST dataset \cite{deng2012mnist}; specifically, we consider both the SVD and t-SNE embeddings \cite{van2008visualizing}. 
Our experiments demonstrate that our algorithms achieve clustering costs that are competitive with standard benchmark algorithms 
that have access to strong oracle queries on the \emph{entire dataset}, while our algorithms only make strong oracle queries on a small fraction of the points (i.e. $1$-$2\%$ of the points). Furthermore, we show that benchmark algorithms with \emph{no} strong oracle queries produce much significantly worse clusterings than our algorithms, demonstrating the necessity of exploiting the strong oracle.

\subsection{Other Related Work}

Our paper studies metric optimization problems where we have easy access to corrupted distances, but accessing the true distances is expensive. This is closely related to 
 both \textit{active learning} and \textit{clustering under budget constraints}, which limit the number of pair-wise comparisons. For the two oracle setting, active learning with both weak and strong labelers
\cite{ZhangC15, YEC20Active}, as well as active learning with diverse labelers
\cite{HuangCMZ17} have been studied. 
In the budget constrained clustering case, a line of work considered spectral clustering on partially sampled matrices \cite{fetaya2015graph,shamir2011spectral,wauthier2012active}, and \cite{garcia2020query} devise correlation clustering algorithms with approximation depending on the query budget. 
Two other closely related lines of work are \textit{clustering with noisy oracles} and \emph{algorithms with predictions} (see \cite{mitzenmacher2022algorithms} for a survey of the latter\footnote{Also see the website \url{https://algorithms-with-predictions.github.io/}} ). Many tasks, including correlation clustering and signed edge prediction \cite{mazumdar2017clustering, mitzenmacher2016predicting, green2020clustering}, $k$-clustering \cite{ashtiani2016clustering,ailon2017approximate,nguyen2023improved,addanki2021design, ergun2022learningaugmented}, 
and MST \cite{ErlebachLMS22,BBFL23MST}, have been studied. 

The key difference between all the aforementioned settings and ours is that they are given immediate access to the true similarities (i.e. $(\cX,d)$ for us), and their noisy queries provide access to the \textit{optimal} clustering (or ground truth labels); for instance, their oracles can be asked queries like ``should $x$ and $y$ be clustered together''?
Comparatively, in our setting the strong oracle simply provides non-noisy access to the input distances. 
For such oracles, perhaps the most closely related work is the recent paper~\cite{silwal2023kwikbucks}, which studies \textit{correlation clustering} with weak and strong oracles akin to ours. However, their model is limited to correlation clustering (where the input is a graph with binary labels), whereas our model is based in a metric space, and thus captures the entire span of metric optimization algorithms. For the setting where we only employ a weak oracle (such as our MST algorithm), perhaps the most closely related work is \cite{mason2019learning}, which studies finding nearest neighbors when distances are corrupted by Gaussian noise, which is an incomparable noise setting to ours. 

  








\section{Preliminaries}
\label{sec:prelim}


A full instance of the weak-strong oracle model is specified by the triple $(\cX,d,\td)$, where $\td:\cX^2 \to \R$ are the distances returned by the weak oracle. We write $\cor \subseteq \binom{n}{2}$ to denote the set of ``corrupted'' distances (where $\td(x,y) \neq d(x,y)$). 
We allow the values of $\td(x,y)$ for $(x,y) \in \cor$ can be chosen arbitrarily and by an adversary who knows the full metric $(\cX,d)$ as well as the set $\cor$. 

We write $\Delta \geq 1$ to denote the \emph{aspect ratio} of the original metric space $(\cX,d)$. 
Without loss of generality (via scaling), we can assume that $1 \leq d(x,y) \leq \Delta$
for all $x,y \in \cX$. Note that this bound only applies to the strong oracle -- the weak oracle distances $\dtilde$ can of course be arbitrarily larger than $\Delta$ or smaller than $1$. Throughout, we will assume that the aspect ratio is polynomially bounded, namely that $\Delta \leq n^c$ for any arbitrarily large constant $c \geq 0$ -- a natural assumption in the literature. We discuss the generalization of our work to arbitrary aspect ratio in \Cref{seubsec:add-detail}.

Our algorithms in Sections \ref{sec:alg-k-center} and \ref{sec:alg-k-means} will only use strong oracle point queries; since any algorithm that makes $\tilde{O}(k)$ point queries can be transformed into an algorithm that makes at most $\tilde{O}(k^2)$ edge queries (simply by querying all stances between the set $S \subset \cX$ of point-queries), the edge query complexity follows as a corollary. 
Thus, in what follows, ``strong oracle query'' refers to strong oracle point queries, and edge queries will be explicitly specified as such. For simplicity, we present our clustering algorithms for the case of $\delta = 1/3$, and describe in Appendix \ref{appendix:deltaRemark} how we can generalize our algorithms for $\delta \in (0, 1/2)$. 

\textbf{Notation.} 
For a set $S \subset \cX$, we write $d(u,S) = \min_{y \in S}d(x,y)$, and we define $\dtilde(u,S)$ analogously. For a metric space $(\cX, d)$, a point $x \in \cX$, and a radius $r >0$, we write $\cB_d(x,r) = \{ y \in \cX \; | \; d(x,y) \leq r \}$ to denote the closed metric ball centered at $x$ with radius $r$ under $d$. When $d$ is the original metric $(\cX,d)$, we simply write $\cB(x,r) = \cB_d(x,r)$.
We write $\OPTcen(d)$, $\OPTmean(d)$, and $\OPTmed(d)$ to denote the optimal clustering cost of $k$-center, $k$-means and $k$-median on $\cX$ with distance metric $d$. 
  When the problem of study is clear, we also simply use $\OPT$ to denote the optimal cost. %
  For the metric minimum spanning tree problem, given a tree $T = (\cX,E)$ spanning the points in $\cX$, we write $w_d(T):=\sum_{(x,y)\in E} d(x,y)$ to denote the cost of the tree in the metric $d$, and $w_{\dtilde}(T):=\sum_{\text{edge} (x,y)\in T} \td(x,y)$ to denote the cost with respects to $\td$. For clarity, we will sometimes write $w(T) = w_d(T)$ and $\tw(T) = w_{\td}(T)$.  

\section{$k$-Center Clustering in the Weak-Strong Oracle Model}
\label{sec:alg-k-center}




We first consider the problem of $k$-center clustering, where the goal is to find $k$ centers $c_1,\dots,c_k \in \cX$ and a mapping $\cC:\cX \to \{c_i\}_{i=1}^k$ such that the maximum distance of a point in $\cX$ to its respective center, namely $\max_{p \in \cX} d(p,\cC(p))$, is minimized. 
We give an algorithm for this problem that achieves $O(1)$-approximate solution using $\widetilde{O}(k)$ strong oracle queries.


\begin{theorem}
\label{thm:alg-k-center}
    For any $\eps > 0$ and metric space $(\cX,d)$, there is an algorithm in the weak-strong oracle model that, with high probability, obtains a $(14+\eps)$-approximation to the k-center problem using $O(k \log^2{n} \cdot \log{\frac{\log{n}}{\eps}})$ strong oracle queries, $O(kn \log^2{n} \cdot \log{\frac{\log{n}}{\eps}})$ weak oracle queries, and running in time $\tilde{O}(n k \log{\frac{1}{\eps}})$. 
\end{theorem}

At a high level, our algorithm is a recursive procedure that, at each step, attempts to correctly cluster (up to a constant approximation) all of the points $p$ whose optimal cluster contains at least a $\frac{1}{10 k}$-fraction of the current set of points. We then remove these clustered points and recurse on the un-clustered points. In what follows, we suppose that we are given a number $R$ such that $2 \OPT \leq R \leq (2+\eps)\OPT$ (we will later find $R$ by guessing it in powers of $(1+\eps)$).  

Let $\cB_1^*,\dots,\cB_k^*$ be the optimal clusters. Our algorithm first samples a set $S$ of $O(k \log n)$ points uniformly from $\cX$ and queries the strong point oracle on $S$; by standard concentration bounds, for any optimal cluster $\cB_i^*$ containing at least a $\frac{1}{10 k}$-fraction of the points, $S$ will contain at least $\Omega(\log n)$ points from $\cB_i^*$ with high probability. We call such a cluster \textit{heavy}. Our goal will be to cluster the points in each heavy ball $\cB_i^*$, remove the clustered points, and recurse on the remaining points. Since at least a $\frac{9}{10}$ fraction of points are in heavy balls, the recursion will complete after $O(\log n)$ iterations.

The main challenge is now to identify the points that belong to a heavy cluster $\cB_i^*$. 
Notice that even if we knew the center $c_i^*$ of $\cB_i^*$, we cannot use  $\td(x,c_i^*)$ alone to determine if a point $x$ should be in $\cB_i^*$, since $\td(x,c_i^*)$ could have been arbitrarily corrupted. Instead, suppose that we were given a sufficiently large set $U \subset \cB_i^*$ of points. Then we observe that one can estimate $d(x,c_i^*)$ by using the \emph{median} of the weak oracle distance between $x$ and the points in $U$. Since each distance $\td(x,u)$ is corrupted with probability $\delta < 1/2$ for $u \in U$, we expect that the median distance should be between $\min_{u \in U} d(x,u)$ and $\max_{u \in U} d(x,u)$.

We now give the formal description and analysis of the algorithm as in \Cref{alg:k-center}. For simplicity, we first state the algorithm using monotonically increasing values of $R$, and we lately describe how this can be sped up via binary searching for a usable value of $R$.

\begin{algorithm}[!htbp]
\caption{\label{alg:k-center}The $k$-center algorithm}
\KwData{Set of points $\cX$, number of centers $k$.}
\KwResult{Clustering $\cC$ with centers $C=\{c_1, \cdots , c_k\}$ and assignment of each point $x \in \cX$}

\For{$R = (1+\eps)^\ell$, $\ell \in [O(\log{n})]$}{
    Run sample and cover (\Cref{alg:sample_cover}) on $\cX$ to obtain candidate centers; \\
    Run greedy ball carving (\Cref{alg:greedy_ball_carve}) on candidate centers with threshold $R$; \\
    If number of centers obtained from the above procedure is $>k$, increase $\ell$ and repeat.
                
}
\end{algorithm}





Before describing the procedures and their guarantees, we first introduce some notations. Let $\cB_1^* = \cB_1^*(c_1^*, \OPT), \cdots , \cB_k^* = \cB_k^*(c_k^*, \OPT)$ centered at $c_1^*, \cdots , c_k^*$ be the clusters corresponding to the optimal k-center solution  on $(\cX,d)$. We define the ``cover'' of a point $x\in \cX$ as follows.

\begin{definition}
\label{def:point-covering}
For any pair of point $(x,y)\in \cX$, we say $y$ is covered by the ball $\cB(x,r)$ centered at point $x$ with radius $r$ if $d(x,y) \le r$ (namely, if $y \in \cB(x,r))$.
\end{definition}


We also define the following notion of ``heavy ball'' that captures the balls with a sufficiently large number of points inside. 
\begin{definition}
\label{def:heavy_ball_k_center}
Let $(\cX,d)$ be a metric space, and fix any $x \in \cX$ and $r > 0$. We say that the ball $\cB(x,r)$ is \underline{heavy} if $\card{\cB(x,r) \cap \cX} \geq  \frac{n}{10k}$. 
\end{definition}

 For $\ell \le k$, denote the corresponding set of \emph{heavy balls} from optimal solution for $k$-center by 
\[\cB^*_\cH = \{\cB_1^*(c_1^*, \OPT), \cdots , \cB_\ell^*(c_\ell^*, \OPT)\}.\]

\begin{observation}
    \label{obs:frac_covering_points}
    The union of all heavy balls from $\cB^*_H$ cover at least $\frac{9n}{10}$  points in $\cX$
\end{observation}
 \begin{proof}
     Consider the balls in optimal solution of $k$-center on $(\cX,d)$ that are not heavy as per \cref{def:heavy_ball_k_center}. There are at most $k-1$ such balls each covering at most $ \frac{n}{10k}$ points in $\cX$ which amounts to at most $n/ 10$ points. As a result, heavy balls cover $\frac{9n}{10}$ points in $\cX$.
 \end{proof}

We define the ``cover'' of set of points by a set of balls as follows.

\begin{definition}
    \label{def:set_covering}
    For a set of points $U$, we say $U$ is covered by $\cB_C =\{\cB(c_1,r_1), \cdots, \cB(c_\ell, r_\ell)\}$, a set of $\ell$ balls, if for all $u \in U$ these exists $\cB(c_i,r_i) \in \cB_C$ such that $\cB(c_i,r_i)$ covers $u$.
\end{definition}

We first describe the procedures ``Greedy Ball Carving'' and ``Sample and Cover'' in greater detail, whose guarantees will ultimately be used to complete the proof for \cref{thm:alg-k-center}. We start by assuming knowledge of an estimate $2\OPT \le R \le 2(1+\eps) \OPT$. In the proof for \cref{thm:alg-k-center} we will guess $R$ in powers of $(1+\eps)$. 

We first discuss the procedure greedy ball carving (\cref{alg:greedy_ball_carve}), which serves as a witness in estimating  $R$. For a set of points with strong oracle queries and some $R$, the procedure first randomly picks a point as center. Removing all points in the set that are at distance $\le R$ from center, it recurses on the remaining points. We will use the centers generated by \cref{alg:greedy_ball_carve} in identification of heavy balls and points within.

\begin{algorithm}[!htbp]
\caption{\label{alg:greedy_ball_carve}Greedy Ball Carving}
\KwData{Set of points $S$, radius $R$}
\KwResult{Set of centers $C=\{c_1, \cdots , c_m\}$ and assignment of $s \in S$ to respective centers.}
\textbf{Init:} $C = \{\}$ \\

\While{$S$ is not empty}{
    Pick arbitrary point $c \in S$ \\ 
    Treating $c$ as center, assign $s$ to $c$ if $d(c,s) \le R$ for all $s \in S$. \\
    Add $c$ to $C$ and remove assigned points.
}
\end{algorithm}

Note that greedy ball carving will \textit{always} be used by our algorithm on sets with strong oracle queries and as a result operates with true distances. The following observation helps us to use greedy ball carving as a witness for guessing $R$, as discussed in \cref{seubsec:add-detail}.


\begin{observation}
\label{obs:certificate_R}
If $\OPT$ is the optimal k-center cost for $(\cX,d)$ and $2\OPT \le R$, then \cref{alg:greedy_ball_carve} run on any subset of points $S \subseteq \cX$, with radius $R$ returns $m \le k$ centers.
\end{observation}

\begin{proof}
The proof is by contradiction. If the number of centers returned by \cref{alg:greedy_ball_carve} is $m > k$, this implies there are $k+1$ points with pairwise distance more than $R$ which contradicts the assumption  $2\OPT \le R$.
\end{proof}

 We now describe the recursive procedure sample and cover (\Cref{alg:sample_cover}) which at each step aims to cluster points in each heavy balls.







\FloatBarrier
\begin{algorithm}[!hbtp]
\caption{\label{alg:sample_cover}The recursive sample and cover procedure}
\KwData{Set of points $\cX$, weak distance oracle $\DQ$, strong oracle $\LQ$, estimate $R$}
\KwResult{Clustering $\cC$ with centers $C=\{c_1, \cdots , c_m\}$ and assignment of each point $c \in \cX$, and $\cX_{\LQ}$ the set of points with $\LQ$ }
\textbf{Init:} $C = \{\}$ and $\cX_{\LQ} = \{\}$  \\

\While{$\cX$ is not empty}{
    Sample $|S| = 100k \log{n}$ and $|T| = 2000k \log{n}$ points u.a.r. from $\cX$ \\
    Query $\LQ$ for $S \cup T$ and set $\cX= \cX \setminus \{S \cup T\}$. \\ 
    \underline{\textbf{Step 1.}} Run \cref{alg:greedy_ball_carve} on $S$ with radius $R$ and let the set of centers obtained be $C$ \\
    Check if $|C| > k$ \\
    \underline{\textbf{Step 2.}} Identify complete balls for $c_i \in C$ using $T$ and add $T \cup \{S \setminus C\}$ to $\cX_{\LQ}$.\\ 
    \For{$x \in \cX \setminus \{S \cup T\}$}{
        \underline{\textbf{Step 3.}} Compute $\dest(x, c_i) = \med{\dtilde(x,y) | y \in T \cap \cB(c_i, 3R)}$ for all complete $\cB(c_i, 3R)$. \\
        \underline{\textbf{Step 4.}} Call $x$ covered by $\cB(c_i,  3R)$ if $\dest(x, c_i) \le 6R$ and assign $x$ to $c_i$.
    }

    Remove all covered points form $\cX$. \\    
    
}


\end{algorithm}
\FloatBarrier




We will prove the approximation guarantee for the set of candidate centers generated by \cref{alg:sample_cover} in the following order, corresponding to respective steps in \cref{alg:sample_cover}. 
\begin{enumerate}[leftmargin=50pt]
    \item[\textbf{Step 1.}] Centers returned by \cref{alg:greedy_ball_carve} on set $S$ can be used to identify heavy balls from $\cB_\cH^*$.
    \item[\textbf{Step 2.}] $T$ provides sufficient points in heavy balls identified in Step 1 for distance estimation.
    \item[\textbf{Step 3.}] Distance estimated using heavy balls is accurate enough for constant factor approximation.
    \item[\textbf{Step 4.}] Using heavy balls and distances from step 3, we cover a constant fraction of points each iteration.
\end{enumerate}

\paragraph{Step 1.} The combined objective of step 1 and 2 in the algorithm is to identify the heavy balls and the points belonging to these heavy balls. We by start sampling two sets of points ($S$ and $T$), each of size $O(k \log{n})$ uniformly at random from $\cX$ and querying strong oracle on them. We first make following observations that help formalize the process of identification of heavy balls.


\begin{observation}
\label{obs:covering_set_S}
Let $n$ be the number of points in $\cX$ and let $S$ be a set of $100 \cdot k \log{n}$ points sampled uniformly at random from $\cX$. Then, with high probability, for every $\cB^*(c_i^*, \OPT) \in \cB^*_H$, there are at least $\log{n}$ points in $S$ covered by $\cB^*(c_i^*, \OPT)$, i.e. $\card{\cB^*(c_i^*, \OPT)\cap S}\geq \log{n}$. 
\end{observation}
\begin{proof}
As each heavy  $\cB_i^*(c_i^*, \OPT) \in \cB^*_H$ covers at least $n/ 10k$ points in $\cX$, in expectation we have at least $10 \log{n}$ points in $S$ that are covered by $\cB_i^*(c_i^*, \OPT)$.
By Chernoff bounds, for $\cB^*_i(c_i^*, \OPT)$ 
\begin{equation*}
    \Pr\paren{\card{\cB_i^*(c_i^*, \OPT)\cap S}\le \log{n}} \le \exp{\paren{-\frac{81}{200}\cdot 10\log{n}}} \le n^{-4}.
\end{equation*}
Using union bound over at most $k$ heavy balls gives us that each $\cB^*(c_i^*, \OPT) \in \cB^*_H$ covers at least $\log{n}$ points in $S$ w.p. $1 - k / n^4$.
\end{proof}

This tells us that $S$ has at least $\log{n}$ points from each heavy ball in optimal solution w.h.p. We now use \cref{alg:greedy_ball_carve} to cover all points in $S$ with at most $k$ centers using radius $R$, and use the centers to form heavy balls. Let ${C}$ be the set of centers that are returned by \cref{alg:greedy_ball_carve} for set $S$. 
For each ${c_i} \in {C}$, consider the balls $\cB({c_i}, 3R)$. 
As $2R + \OPT \le 3R$, using $3R$ as radius for $\cB({c_i}, 3R)$ ensures that the union of these balls cover all points that are covered by $\cB_H^*$ ($\frac{9n}{10}$ points in $\cX$). 
We now define $\cB_\cH$ as the set of $\cB({c_i}, 3R)$ that are heavy by \cref{def:heavy_ball_k_center} among all ${c_i} \in {C}$, i.e.,
\begin{equation}
\label{eq:heavy_ball_size}
\cB_H = \{ \cB({c_i}, 3R) \, \mid\, c_i\in C, \, \card{\cB({c_i}, 3R)}\geq  \frac{n}{10k} \}.
\end{equation}
Following similar argument as \Cref{obs:frac_covering_points}, non-heavy balls $\cB(c_i, 3R)$ covers at most $n/ 10$ points in $\cB_H^*$, which equates to $\cB_\cH$ covering at least $(\frac{9n}{10} - \frac{n}{10}) = \frac{8n}{10}$ points in $\cX$. 






\paragraph{\textbf{Step 2.}} We now argue that $T$ has enough points for each heavy ball in $\cB_H$, which can be used for distance estimation for the remainder of points in $\cX$. We show that it suffices to have $O(\log{n})$ points in $T$ that are covered by a heavy ball to get a good estimate to its center. Following this, we call $\cB({c_i}, 3R)$ `\underline{\textit{complete}}' if $T$ contains at least $100\log{n}$ points that are covered by $\cB({c_i}, 3R)$. Note that this can be verified by the algorithm as we use strong oracle queries for $S$ and $T$. If any of the points from $S$ or $T$ are uncovered, we just add them to the list of candidate centers (complete balls centers), and since we run greedy ball carving at the end on candidate centers they will be covered. We start by proving the following claim.


\begin{claim}
\label{claim:big_complete} For the set of centers ${C}$ returned by \cref{alg:greedy_ball_carve} on $S$ and $|T| = 2000 \cdot k\log{n}$ set of points sampled uniformly at random from $\cX$, every $\cB({c_i}, 3R) \in \cB_{H}$ for $c_i \in C$ is \textit{complete} w.h.p.
\end{claim}
\begin{proof}
    The argument is similar to \Cref{obs:covering_set_S}, where we first show for a particular ${c_i} \in {C}$ and then union bound over at most $k$ such centers. From \cref{eq:heavy_ball_size}, in expectation $T$ has at least $200 \log{n}$ points covered by a heavy $\cB({c_i}, 3R)$. Using Chernoff bounds, 
\begin{equation*}
    \Pr\paren{\card{\cB_i(c_i, 3R)\cap T} \le 100\log{n}} \le \exp{\paren{-\frac{200}{8}\log{n}}} \le n^{-25}.
\end{equation*}
    Using union bound over at most $k$ centers, w.p. $1 - (k/n^{25})$ every heavy $\cB({c_i}, 3R)$ is complete.
\end{proof}

\paragraph{\textbf{Step 3.}} At this stage we have at most $k$ complete balls $\cB({c_i}, 3R)$, each with $O(\log{n})$ points.
We now use median of weak oracle queries to these points in each heavy ball for distance estimation of remaining points $x \in \cX \setminus {S \cup T}$.
For each $x \in \cX \setminus \{S \cup T\}$ and each complete $\cB({c_i}, 3R)$, let $\dest(x, {c_i}) = \med{\dtilde(x,q) | q \in T \cap \cB({c_i}, 3R)}$ where $\dtilde(v,q)$ denotes the weak oracle queries. These distance estimates form the `workhorse' of sample and cover for assigning points to heavy balls. The following lemma quantifies the accuracy of the estimated distance.


\begin{lemma}
\label{lem:median_dist_est}
    Fix any point $u \in \cX$, radius $r\geq 0$, and set $U \subset  \cB(u,r)$ with $|U| = \Omega(\log n)$. For any $x \in \cX$, define $\dest(x, u) = \med{\dtilde(x,q) | q \in U}$. Then with high probability, for all $x \in \cX$ we have $| \dest(x, u) - d(x, u) | \le r$.
\end{lemma}

\begin{proof}
    Recall that the weak oracle returns arbitrary value for $\dtilde(x,y)$ w.p. $1/3$, independently for each $x,y \in \cX$. For $|U| \ge c\cdot \log{n}$ for a sufficiently large constant $c$, in expectation $\frac{2}{3} c \cdot\log n$ weak oracle calls are not corrupted. For using median as an estimate, we only need more than half queries to not be corrupted. Let $\mathbf{X}_{u}$ denote the number of uncorrupted weak oracle queries $U$. Using Chernoff bounds,
\begin{align*}
    \Pr\paren{\mathbf{X}_{u} \le \frac{c}{2}\log{n}} \le \exp\paren{-\frac{c \log{n}}{16 \cdot 3}} \le n^{-c'}
\end{align*}  

Thus, for sufficiently large $c$, with high probability the median distance lies between $\min_{u \in U} d(x,u)$ and $\max_{u \in U} d(x,u)$. It follows from union bound over at most $n$ points and triangle inequality that $| \dest(x, u) - d(x, u) | \le r$ for all $x \in \cX$.

\end{proof}

Considering \Cref{lem:median_dist_est} for estimating distance to centers of complete balls, for each complete $\cB(c_i, 3R)$  we have $200 \cdot \log{n}$ points in $T$. For a particular complete $\cB(c_i, 3R)$ we have w.p. $n^{-4}$, $| \dest(x, c_i) - d(x, c_i) | \le 3R$ for $x \in \cX \setminus \{S \cup T\}$. Taking union bound over at most $k$ complete balls and distance estimates for at most $n$ points, the estimation guarantee holds w.p. $1 - \frac{k}{n^3}$ for all distance estimates of $x \in \cX \setminus \{S \cup T\}$ to all complete balls $\cB(c_i, 3R)$.

\paragraph{\textbf{Step 4.}} With the distance estimates from \Cref{lem:median_dist_est}, the algorithm assigns $x$ to a complete $\cB(c_i, 3R)$ if $\dest(x,c_i) \le 6R$, with ties broken arbitrarily. Recall that union of all complete $\cB(c_i,3R)$ balls cover $8/10$ fraction of points in $\cX$. As the distance estimates are accurate upto $\pm 3R$ and we assign if distance to centers is at most $6R$, at each step $8/10$ fraction of points are covered by the complete balls. This is achieved by making $O(k \log{n})$ strong oracle queries (recall the algorithm uses strong oracle queries only for set $S$ and $T$). 

Thus, at each iteration we cover a constant fraction of points in $\cX$ using at most $k$ centers. For covering all points in $\cX$, the recursion will proceed for $O(\log n)$ rounds, generating $O(k \log n)$ centers. We now prove the approximation guarantee obtained by these set of centers.

\begin{lemma}
\label{lem:sample-cover}
Suppose $2\OPT\leq R\leq 2\cdot (1+\eps)\OPT$, then with high probability, \Cref{alg:sample_cover} computes a set of $O(k\log{n})$ centers and the clustering assignment for each point $x\in \cX$ using $O(k \log^2 n)$ strong oracle queries, $O(nk \log^2 n)$ weak oracle queries such that each point is at most $12 \cdot (1+\eps) \OPT$ from its respective center.
\end{lemma}

\begin{proof}

At each iteration, \cref{alg:sample_cover} produces at most $k$ centers that cover constant fraction of points in $\cX$ and the procedure runs for $O(\log{n})$ iterations. The total number of strong oracle queries we make is $(\log n \cdot (\card{S} +\card{T})) = O(k\log^2 n)$. For number of weak oracle queries, note at each iteration we have $O(\log n)$ points in at most $k$ complete balls. For estimating distances of at most $n$ remaining points at each iteration we make $O(nk\log n)$ weak oracles queries. For $O(\log n)$ such iterations, we make total $O(nk \log^2 n)$ weak oracle queries.

At the end of recursion, any point $x \in \cX$ is at most $6R$ from respective center in set of $O(k\log n)$ centers, and $2\OPT \le R \le 2 \cdot (1+\eps) \OPT$, this gives us $12 \cdot (1+\eps)$-approximate solution.

\end{proof}


In order to obtain a solution for $k$-center problem, we now run \cref{alg:greedy_ball_carve} on the set of $O(k \log n)$ centers (also uncovered points with strong oracle query, if any, from $S$ and $T$ at any iteration).

From \Cref{obs:certificate_R} we obtain at most $k$ centers. Using the approximation guarantee from \Cref{lem:sample-cover} and additional iterations for estimation of $R$, we now wrap up the proof for \cref{thm:alg-k-center}.

\begin{proof}[Proof of \cref{thm:alg-k-center}]
We first analyze the total number of strong and weak oracle queries used by the algorithm, by looking at additional overhead from estimating $R$. First, while there are $O(\frac{1}{\eps} \log n)$ possible values of $R$ to try, instead of trying each in increasing order as in \cref{thm:alg-k-center}, we claim that we can binary search to find the correct value. Specifically, for the range $r \in \{1,(1+\eps),(1+\eps)^2,\dots,\Delta\}$, we simply need to find a value of $R$ such that running \cref{alg:k-center} on that value of $R/(1+\eps)$ returns more than $k$ centers, (thus implying that $R < 2 (1+\eps) \OPT$), and such that running \cref{alg:k-center} on $R$ returns at most $k$ centers (thus implying that we get a valid solution from \cref{alg:k-center}  with a value of $R$). Importantly, this need not be the smallest $R$ such that the above occurs. Thus, we can find such an $R$ via binary search, which requires a total of $O( \log (\frac{\log n}{\eps}))$ rounds of running \cref{alg:k-center}. Since each run of \cref{alg:k-center} uses  $O(k \log^2 n )$ strong oracle queries and $O(nk \log^2 n)$ weak oracle queries, after the binary search it follows that the total query  complexity is $O(k \log^2 n \log \left(\frac{\log n}{\eps}\right))$ strong oracle queries and $O(nk \log^2 n\log \left(\frac{\log n}{\eps}\right))$ weak oracle queries. 



Now we look at the approximation factor. Let $R = (1+\eps)^j$ be such a value found via the above binary search. As above, we have that $R \le 2(1+\eps)\OPT$, and that we obtained  a valid solution from \cref{alg:k-center} with a value of $R$. For covering the points, since we assign points to a complete $\cB(c_i, 3R)$ with distance to center at most $6R$ and later run \cref{alg:greedy_ball_carve} with threshold $R$ on $O(k \log{n})$ centers, any point is at most $7R$ away from the center which gives a $14(1+\eps)$-approximate solution (since $R \leq 2(1+\eps)\OPT$). 

For analyzing the run-time of our algorithm, we first look at one iteration of our algorithm. Note, greedy ball carving on a set of $O(k \log n)$ points for at most $k$ centers takes $O(k^2 \cdot \polylog n)$ time. Then, conditioning on the high probability event of \Cref{claim:big_complete}, for at most $k$ complete balls with $O(\log n)$ points each, estimating median of remaining points to respective centers takes $O(n k \cdot \polylog n)$ time. At most $O(\log n)$ such iterations coupled with iterations for estimating $R$, the total run time is $O\left(k^2 \cdot \polylog n \log \left(\frac{\log n}{\eps}\right) + n k \cdot \polylog n \log \left(\frac{\log n}{\eps}\right)\right) = \tilde{O}\left(n k \log \left(\frac{1}{\eps}\right)\right)$.

\myqed{\Cref{thm:alg-k-center}}
\end{proof}

\newcommand{\AijHeavy}{\ensuremath{\mathcal{E}_{A_{i,j}\text{-heavy}}}\xspace}
\newcommand{\SijL}{\ensuremath{S^{\text{light}}_{i,j}}\xspace}
\newcommand{\SijH}{\ensuremath{S^{\text{heavy}}_{i,j}}\xspace}

\section{$k$-Means and $k$-Median Clustering in the Weak-Strong Oracle Model}
\label{sec:alg-k-means}
We now consider algorithms for the $k$-means and $k$-median clustering problems, which will differ significantly from our $k$-centers algorithm. 
Our main algorithmic result is as follows.
\begin{theorem}
\label{thm:alg-k-means}
    There exists an algorithm that given a metric space $(\cX, d)$ and the oracles under the weak-strong oracle model, with high probability computes an $O(1)$-approximate solution to the $k$-means and $k$-median problems using $O(k \log^{2}{n})$ strong oracle queries and $O(nk\log^{2}{n})$ weak oracle queries. Furthermore, the algorithms runs in $O((nk+k^3)\cdot \polylog{n})$ time.
\end{theorem}

The algorithm proceeds by first constructing a \textit{coreset} $S \subset \cX$ of at most $O(k\log^2{n})$ points, which contains a set of $k$ centers that are $O(1)$-approximation to $\OPT$. 
To build the coreset $S$, we arbitrarily order the points, and for point $p$ (in order), we sample $p$ into $S$  with probability proportional the the distance $d(p,S)$. Since the weak-oracle distances can be arbitrarily corrupted, we cannot naively use the weak oracle to compute the $d(p,S)$. Instead, we design a \emph{proxy distance} which we will use to estimate $d(p,S)$, based on taking medians of weak-oracle distances from $p$ to pair-wise close subsets of points in $S$. Specifically, we define the ``heavy balls distance'' as follows:

\begin{definition}[Heavy-ball nearest distance]
\label{def:heavy-ball-dist}

Let $\cX$ be a set of points and $d$ be the underlying metric. Furthermore, let $y\in \cX$ be a point and $S\subseteq \cX$ be a set of points such that $y \not\in S$ and $\card{S}\geq 100\log{n}$. For any vertex $x \in S$ and radius $r > 0$, we set  $Q(x, S, y, r) = \med{\dtilde(z, y) \mid z \in S,  d(x, z) \leq r } + 6 \cdot r$, and define the \underline{heavy-ball nearest distance} as
\begin{align*}
Q(y, S) = \min_{\substack{x \in S, r > 0 \\ \card{\cB(x,r) \cap S} \geq 100 \log(n) }} Q(x, S, y, r).
\end{align*}
\end{definition}
We prove that, with high probability, we always have $Q(y,s) \geq d(y,S)$. 
In order to compute the heavy-ball distance, we will need to know the distance between all pairs of points within $S$, thus we will query the strong oracle on all points added to the coreset $S$. 
Our full $k$-means (and $k$-median) algorithm is then presented in in \Cref{alg:k-means}. We demonstrate that, for the correct guess of $\widetilde{\OPT}$, the algorithm samples at most $O(k\log^2{n})$ points (therefore bounding the query complexity), and that this coreset contains a set of $k$-centers which are a $O(1)$-approximation for the  $k$-means (or $k$-medians) objective. We can then run a point-weighted $k$-clustering algoithm on the coreset $S$ to obtain this $O(1)$-approximate solution.



\begin{algorithm}[!t]
\caption{\label{alg:k-means}The $k$-means (and $k$-median) algorithm}
\KwData{Set of points $\cX$; weak oracle $\DQ$, strong oracle $\LQ$; estimation of the optimal cost $\widetilde{\OPT}$}
\KwResult{A clustering $\mathcal{C}$, which includes a set of centers $C=\{c_1, \cdots , c_m\}$ and the assignment of each point $x\in \cX$}
\textbf{Init:} Label the points as $\{x_{1},\cdots, x_{n}\}$ by an arbitrary order; let $S = \{x_{1}\}$; a counter $w(x_{1})=1$. \\
Compute the value $f=\frac{1}{20}\cdot \frac{\widetilde{\OPT}}{k\log^2{n}}$. \\
\For{$x_{i} \in \cX$}{
    \eIf{$\card{S}<100\log{n}$}{
        Add $x_{i}$ to $S$ and query the strong oracle $\LQ$ for $x_{i}$. 
    }{
    Sample $x_{i}$ with probability $\min\{1, Q(x_{i}, S)/f\}$.\\
    \eIf{$x_{i}$ is sampled}{
    Add $x_i$ to $S$, i.e. $S\leftarrow S\cup \{x_{i}\}$, make the strong oracle query $\LQ$ on $x_{i}$, and add to the counters $w(x_{i})=1$.\\
    }{
    Assign $x_{i}$ to the cluster induced by $x'$, which is defined as the center of the ball that attains the heavy-ball nearest distance as in \Cref{def:heavy-ball-dist}. \\
    Increase $w(x')$ by 1.
    }
}
}
Run a weighted $k$-means (resp. $k$-median) clustering algorithm on $S$, e.g., algorithms of \cite{MettuP04,AhmadianNSW17}.
\end{algorithm}

\paragraph{The analysis.} We now proceed to analyze this algorithm and evetually prove \Cref{thm:alg-k-means}.
\Cref{alg:k-means} requires an approximation $\widetilde{\OPT}$,of the optimal cost $\OPT$. We will first assume we have such a value $\widetilde{\OPT}$ satisfying  $2^i \leq \widetilde{OPT}\leq 2^{i+1}$, and later describe how we can find it via binary search. Specifically, we will terminate any run of \Cref{alg:k-means} whenever it samples more than $1800 k \log^2{n}$ points to add to $S$, and conduct a binary search by maintaining upper and lower bounds of the indices $i$. In the end, we output the subroutine induced by an $\istar$ such that $(i).$ the subroutine with $2^{\istar}$ returns a clustering (i.e. did not terminate) and $(ii).$ the subroutine with $2^{\istar-1}$ is terminated without return. We then set $\widetilde{\OPT} = 2^{\istar}$. We will show that this value of $\widetilde{\OPT}$ satisfies the desired properties. 
As such, our analysis focus on the case when $\OPT \leq \widetilde{\OPT}\leq 2\OPT$.

For the simplicity of presentation, we focus on the analysis of $k$-median as it does not involved the square terms on the distances. We show in \Cref{rmk:analysis-on-k-means} how the analysis also works for the $k$-means clustering with a slightly larger constant factor.



To proceed with the analysis, we also introduce some self-contained notations used in this analysis. We let $\cstar_{1},\cdots, \cstar_{k}$ be the centers for the optimal $k$-median solution, and we let $\rstar_1, \rstar_2, \cdots, \rstar_{k}$ be the \emph{average} cost of the points in each cluster, i.e.,
\[\rstar_{i} = \frac{1}{\card{\{x\mid \cC(x)=\cstar_{i}\}}}\cdot \sum_{x \, s.t.\, \cC(x)=\cstar_{i}}d(x, \cstar_{i}).\]
Based on this, we can define $B^{j}_{i}$ as the ball centered at $\cstar_{i}$ and with distance at most $2^{j} \rstar_{i}$. We further define $A_{i}^{j}$ to be the set of points with distance $(2^{j-1}\rstar_{i}, \,\, 2^{j} \rstar_{i}]$ from the center of $C_{i}$, i.e. an ``annulus'' between $B^{j-1}_{i}$ and $B^{j}_{i}$. We also define the following event:
\begin{quote}
$\AijHeavy$: the set $A^{j}_{i}$ has at least $100\log{n}$ points in $S$, i.e. $\card{A^{j}_{i} \cap S}\geq 100\log{n}$.
\end{quote}
\noindent
We let the points being sampled in $S$ before $\AijHeavy$ (also denote as $\neg \AijHeavy$) as $\SijL$, and let its complement be $\SijH$, i.e. the set of points sampled after $\AijHeavy$. In what follows, we will first show the approximation guarantees and the number of centers in $S$.

\begin{lemma}
\label{lem:k-means-S-correct}
Suppose $\OPTmed \leq \widetilde{\OPT}\leq 2\OPTmed$. Then, with probability at least $1-\frac{1}{n^3}$, the clustering outputted by $S$ in \Cref{alg:k-means} gives an $O(1)$-approximation of the $\OPTmed$.
\end{lemma}
\begin{proof}
For each of the set $A^{j}_{i}$, we analyze the cost it pays in the formation of $S$ by looking at $\neg \AijHeavy$ (before $A^{j}_{i}$ becomes heavy) and $\AijHeavy$ (after $A^{j}_{i}$ becomes heavy), respectively as in \Cref{clm:k-means-cost-light} and \Cref{clm:k-means-cost-heavy}.
\begin{claim}
\label{clm:k-means-cost-light}
For each set $A^{j}_{i}$, with probability at least $1-\frac{1}{n^4}$, the cost induced by all points in $A^{j}_{i}\cap \SijL$ is at most $200 f \cdot \log{n}$.
\end{claim}
\begin{proof}[Proof of \Cref{clm:k-means-cost-light}]
For any point $x \in A^{j}_{i}\cap \SijL$, we define random variable $X_{x}$ as the indicator random variable for $x$ to be sampled. Since we sample each point with probability at most $Q(x, S)/f$, there is
\begin{align*}
\expect{X_{x}} \leq Q(x, S_{:x})/f,
\end{align*}
where we use $S_{:x}$ to denote the sampled set before $x$ is visited.
Suppose $\AijHeavy$ happens after sampling $N$ points from $A^{j}_{i}$, which means up till the $N-1$ point, we have
\begin{align*}
\expect{\sum_{i\in [N-1]} X_{x_{i}}} \leq \sum_{i\in [N-1]} Q(x, S_{:x_{i}})/f < 100 \log{n}.
\end{align*}
The last inequality holds since we condition on the event that the number of sampled points is (deterministically) less than $100 \log{n}$. Similarly, we can also get a lower bound on the expectation by using the fact that we are only one point short of reaching $100 \log{n}$ (assuming $n\geq 10$).
\begin{align*}
\expect{\sum_{i\in [N-1]} X_{x_{i}}} \geq 99 \log{n}.
\end{align*}

On the other hand, note that if a point $X_{x}$ is not sampled, the cost induced by $X_{x}$ is exactly $Q(x, S_{:x_{i}})$. As such, for the induced cost of the points in $ A^{j}_{i}\cap \SijL$ to be more than $200 f\cdot \log{n}$, there must be
\begin{align*}
\sum_{i\in [N-1]} Q(x, S_{:x_{i}}) > 200 f\cdot \log{n}.
\end{align*}
Comparing the two inequalities, and using the fact that $\sum_{i\in [N-1]} X_{x_{i}}$ is a summation of 0/1 random variables, we have that for the second inequality to happen, the probability is
\begin{align*}
\Pr\paren{\text{cost induced by $A^{j}_{i}\cap \SijL>200 f\cdot \log{n}$}} &= \Pr\paren{\sum_{i\in [N-1]} Q(x, S_{:x_{i}}) > 200 f\cdot \log{n}}\\
&\leq \Pr\paren{\sum_{i\in [N-1]} X_{x_{i}} \geq 2\cdot \expect{\sum_{i\in [N-1]} X_{x_{i}}}}\\
&\leq \exp\paren{-\frac{1}{3}\cdot \expect{\sum_{i\in [N-1]} X_{x_{i}}}} \tag{by multiplicative Chernoff}\\
&\leq \frac{1}{n^7},
\end{align*}
as desired. \myqed{\Cref{clm:k-means-cost-light}}
\end{proof}

We now turn to the cost of the points in $A^{j}_{i}$ after the set becomes heavy, i.e. conditioning on $\AijHeavy$ happens.
\begin{claim}
\label{clm:k-means-cost-heavy}
For each set $A^{j}_{i}$, with probability at least $1-\frac{1}{n^4}$, the cost induced by all points in $A^{j}_{i}\cap \SijH$ is at most $13$-multiplicative of the cost induced by $A^{j}_{i}\cap \SijH$ on the optimal clustering.
\end{claim}
\begin{proof}
Note that the moment $\AijHeavy$ happens, we have the ball $B^{j}_{i}$ becomes heavy as well. As such, by the same argument we used in \Cref{lem:median_dist_est}, the cost of adding any $x\in A^{j}_{i}\cap \SijH$ to its nearest heavy-ball cluster is at most $Q(x, S)\leq d(x,\cstar_{i})+6\cdot 2^{j}\rstar_{i}$ with probability at least $1-\frac{1}{n^4}$. Furthermore, we have $d(x,\cstar_{i})\geq 2^{j-1}\rstar_{i}$ by the fact that $x\in A^{j}_{i}$. Therefore, we can charge the cost of $x$ induced in $S$ to at most $13\cdot  d(x,\cstar_{i})$. This implies a $13$-multiplicative approximation for all points in $A^{j}_{i}\cap \SijH$, as desired. \myqed{\Cref{clm:k-means-cost-heavy}} 
\end{proof}

We now finalize the proof of \Cref{lem:k-means-S-correct}. We first argue that we can apply \Cref{clm:k-means-cost-light} to \emph{all} $A^{j}_{i}$ sets. Note that there are at most $O(k\log{n}) \leq O(n\log{n})$ such sets, as for any cluster $\cstar_{i}$, there is no point whose distance is more than $n \, \rstar_{i}$. Therefore, we can apply a union bound to conclude that with probability at least $1-\frac{1}{n^5}$, the bound of \Cref{clm:k-means-cost-light} applies to all such $A^{j}_{i}$ sets. Collectively, they induce at most $200\cdot fk\log^2{n}\leq 10 \, \widetilde{OPT}\leq 20\, \OPTmed$. Furthermore, by \Cref{clm:k-means-cost-heavy}, the total cost we induce on $A^{j}_{i}\cap \SijH$ for all $i,j$ is at most $13\OPTmed$. As such, the total cost induced by $S$ is at most $O(1)\cdot \OPTmed$, as desired.
\end{proof}

We now turn to the analysis of the number of centers that are ever sampled in $S$.
\begin{lemma}
\label{lem:k-means-S-sample}
Suppose $\OPTmed \leq \widetilde{\OPT}\leq 2\OPTmed$. Then, with probability at least $1-\frac{1}{n^3}$, the set $S$ output in \Cref{alg:k-means} has at most $O(k\log^2{n})$ points.
\end{lemma}
\begin{proof}
We again analyze the number of sampled points before and after the event $\AijHeavy$. Note that before $\AijHeavy$, we deterministically only provide at most $100\log{n}$ centers from $A^{j}_{i}$. After $\AijHeavy$, the ball $B^{j}_{i}$ becomes heavy. As such, define $X_{x}$ be the indicator random variable for sampling $x$ from $A^{j}_{i}$, we can condition on the high probability event of \Cref{clm:k-means-cost-heavy}, and argue that with probability at least $1-\frac{1}{n^3}$, there is
\begin{align*}
\expect{X_{x} \mid \AijHeavy} &\leq 260 \cdot \frac{d(x, \cstar_{i}) \, k \, \log^{2}{n}}{\widetilde{\OPT}}\\
&\leq 260 \cdot \frac{d(x, \cstar_{i}) \, k \, \log^{2}{n}}{\OPTmed}.
\end{align*}
Therefore, we can add up all the centers and their $j$ values, which gives us
\begin{align*}
\expect{\sum_{i,j}\sum_{x \in A^{i}_{j}} X_{x} \mid \AijHeavy} &= \sum_{i,j}\sum_{x \in A^{i}_{j}} \expect{X_{x} \mid \AijHeavy} \tag{linearity of expectation}\\
&\leq 260\cdot  \frac{k\log^2{n}}{\OPTmed} \cdot \sum_{i,j}\sum_{x \in A^{i}_{j}} d(x, \cstar_{i})\\
&= 260\cdot k\log^2{n}. \tag{$\sum_{i,j}\sum_{x \in A^{i}_{j}} d(x, \cstar_{i})=\OPTmed$}
\end{align*}
As such, we can bound the \emph{expectation} of the total number of points that are ever sampled as:
\begin{align*}
\expect{\sum_{i,j}\sum_{x \in A^{i}_{j}} X_{x}} & \leq \expect{\sum_{i,j}\sum_{x \in A^{i}_{j}} X_{x} \mid \AijHeavy} +  \expect{\sum_{i,j}\sum_{x \in A^{i}_{j}} X_{x} \mid \neg \, \AijHeavy}\\
&\leq  260\cdot k\log^2{n} + \sum_{i,j} 100 \log{n}\\
&\leq 360\cdot k \log^2 {n}. \tag{there are only $k\log{n}$ possible $A^{j}_{i}$ sets}
\end{align*}
Finally, we analyze the concentration of the number of sampled points in $S$. Note that using \Cref{fct:neg-cor-rvs}, we can verify that the $X_{x}$ indicator random variables are negatively correlated, i.e. $\Pr(X_{x}=1\mid X_{u}=1)\leq \Pr(X_{x}=1)$ for $u\neq x$. If $\sum_{i,j}\sum_{x \in A^{i}_{j}} X_{x}< 100 k \log^2 {n}$, we deterministically obtain the desired bound; on the other hand, assuming $\sum_{i,j}\sum_{x \in A^{i}_{j}} X_{x}\geq 100 k \log^2 {n}$, we have
\begin{align*}
\Pr\paren{\sum_{i,j}\sum_{x \in A^{i}_{j}} X_{x} \geq 5\cdot \expect{\sum_{i,j}\sum_{x \in A^{i}_{j}} X_{x}}} &\leq \exp\paren{-\frac{16}{6}\cdot \expect{\sum_{i,j}\sum_{x \in A^{i}_{j}} X_{x}}} \tag{Chernoff bound for negatively correlated random variables, \Cref{prop:chernoff-general}}\\
&\leq \frac{1}{n^3},
\end{align*}
as desired.
\end{proof}
By \Cref{lem:k-means-S-correct} and \Cref{lem:k-means-S-sample}, after the procedure that produces $S$ terminates, we have set set of at most $O(k \log^2{n})$ centers and an $O(1)$-approximation of the optimal k-median objective. We now describe how this result leads to the $k$-median solution with exactly $k$ centers. To this end, we introduce the following standard result for approximating $k$-median on coresets (see e.g. \cite{GuhaMMO00}) and a proof for completeness.

\begin{proposition}[\cite{GuhaMMO00}]
\label{prop:coreset-approx-original}
Let $\{c_{1},\cdots, c_{m}\}$ be a set of centers that achieves $\alpha$-approximation of the $k$-median (resp. $k$-means) objective, and let $w_{1}, \cdots, w_{m}$ be the number of points contained in each cluster induced by $\{c_{i}\}_{i=1}^{m}$. Furthermore, let $\{\ctilde_{1}, \ctilde_{2}, \cdots, \ctilde_{k}\}$ be a $\beta$-approximate $k$-median (resp. $k$-means) on the weighted points $\{c_{1},\cdots, c_{m}\}$. Then, the clustering induced by $\{\ctilde_{1}, \ctilde_{2}, \cdots, \ctilde_{k}\}$ gives an $O(\alpha+\beta)$-approximation of $\OPTmed$ (resp. $\OPTmean$).
\end{proposition}
\begin{proof}
For any point $x\in \cX$, let $c(x)$ be its clustering center in $\{c_{i}\}_{i=1}^{m}$ and $\ctilde(x)$ be its clustering center in $\{\ctilde_{j}\}_{j=1}^{k}$. If $c(x)=\ctilde(x)$, the cost induced by $x$ remains $d(x, c(x))$. On the other hand, if $c(x)\neq\ctilde(x)$, the cost induced by $x$ is at most $d(c(x), \ctilde(x))+d(x, c(x))$. Furthermore, note that if we let $\OPT_{ws}$ be optimal cost of the weighted cost of clustering on $\{\ctilde_{j}\}_{j=1}^{k}$, and let $\ctilde(\cdot)$ and $\cstar(\cdot)$ be the functions that maps 1). points in $\{c_{i}\}_{i=1}^{m}$ to the center in the weighted clustering and 2). points in $\cX$ to the optimal $k$-median clustering, respectively. As such, we have
\begin{align*}
\OPT_{ws} &= \sum_{i=1}^{m} w_{i} \cdot d(c_{i}, \ctilde(c_{i}))\\
& =\sum_{i=1}^{n} d(c_{i}, \ctilde(c_{i})) \tag{duplicating each $c_{i}$ for $w_{i}$ times and map all of them to the corresponding center}\\
&\leq \sum_{i=1}^{n} d(c_{i}, \cstar(c_{i}))+ d(\ctilde(c_{i}), \cstar(\ctilde(c_{i})))\\
&\leq 2\OPTmed.
\end{align*}
Therefore, we can bound the total cost with
\begin{align*}
\sum_{x\in \cX} d(x, \ctilde(x)) & \leq \sum_{x\in \cX} d(x, c(x)) + \sum_{x\in \cX} d(c(x), \ctilde(x))\\
&\leq \alpha \cdot \OPTmed + \beta\cdot \OPT_{ws}\\
&\leq (\alpha+2\beta)\cdot \OPTmed,
\end{align*}
as desired. 

Finally, for the $k$-means case, we need to replace $d(\cdot, \cdot)$ with $d^{2}(\cdot, \cdot)$; although this is no longer a metric, we can use the approximate triangle inequality that
\begin{align*}
d^{2}(x,z)\leq 2(d^{2}(x,y)+d^{2}(y,z)).
\end{align*}
As such, by replacing the $k$-median clustering centers with the $k$-means ones and use $d^{2}(\cdot, \cdot)$, we can get
\begin{align*}
\sum_{x\in \cX} d^{2}(x, \ctilde(x)) & \leq 2\cdot \sum_{x\in \cX} d^{2}(x, c(x)) + 2\cdot \sum_{x\in \cX} d^{2}(c(x), \ctilde(x))\\
&\leq (2\alpha+8\beta)\cdot \OPTmean,
\end{align*}
as desired.
\end{proof}


\begin{proof}[Proof of \Cref{thm:alg-k-means}]
Conditioning on the right guess of $\OPTmed \leq \widetilde{\OPT}\leq 2\OPTmed$, by \Cref{lem:k-means-S-correct} and \Cref{lem:k-means-S-sample}, in the construction of set $S$ we sample at most $O(k\log^2{n})$ points and produce an $O(1)$-approximation. This also implies we make at most $O(k\log^2{n})$ strong oracle $\LQ$ queries. To bound the number of weak oracle queries, note that in each iteration of $x_i$, we only need to query the distances between the $x_{i}$ and the points in $S$, which is at most $O(nk\, \log^2{n})$, and our post-processing of the points in $S$ does not involve any additional oracle queries.

To see the time efficiency, observe that we can maintain a heap for each point in $S$ to represent the balls of size $O(\log{n})$. As such, when $x_{i}$ is sampled and a strong oracle query $\LQ$ is added, it takes $O(\card{S}\log{n})$ time to insert the value and remove the largest one from the heap. Therefore, conditioning on the high probability event of \Cref{lem:k-means-S-sample}, the updates of the heavy balls takes at most $O(\sum_{\card{S}=1}^{k\log^2{n}} \card{S}\log{n})=O(k^2 \, \polylog{n})$ time. On the other hand, if $x_{i}$ is not sampled, we need to estimate the median from $x_{i}$ to every ball in $S$, and it take $O(\card{S}\log{n})$ time by the heap structure. Therefore, the estimation of heavy ball nearest distance takes $O(n\, \card{S}\log{n})\leq O(nk \, \polylog{n})$ time across the process. Finally, we can run $O(1)$-approximate $k$-median algorithms in $O(\card{S}k)$ time by \cite{MettuP04}, which gives us $O(k^2\polylog{n})$ time by the size bound of $\card{S}$. In total, this gives us the desired $O(nk\, \polylog{n})$ time. 

For the approximation guarantee, by \Cref{lem:k-means-S-correct} and \Cref{prop:coreset-approx-original}, we can run an $O(1)$-approximate $k$-median (resp. $k$-means) algorithm on the $O(1)$-approximate coreset of $S$ \footnote{If we do not care about the time efficiency, we can run the state-of-the-art polynomial-time algorithm by \cite{AhmadianNSW17}; we can even run a brute-force algorithm to search for the minimum-cost clustering}. The correctness is guaranteed since we know the \emph{exact} distance between the points in $S$ (by the $\LQ$ queries). Therefore, we get an $O(1)$ approximation of the optimal clustering cost.

Finally, when running with unknown $\widetilde{OPT}$, we break a run whenever it samples more than $1800 k \log^2{n}$ points in the construction of $S$, and output the run with the $\widetilde{OPT}$ value that $(i)$. is not break and $(ii).$ is next to a run with $\widetilde{OPT}/2$ that breaks. Thus, we can binary search for such a value of $\widetilde{OPT}$ (as discussed in \Cref{seubsec:add-detail}).  Using that $\Delta \leq \poly(n)$, 
this gives us $O(\log\log{n})$ overhead for the queries and running time. As such, it results in at most $O(k\log^2{n}\log\log{n})$ strong oracle $\LQ$ queries, at most $O(nk\log^2{n}\log\log{n})$ weak oracle $\DQ$ queries, and $O(nk\, \polylog{n})$ time. 
\myqed{\Cref{thm:alg-k-means}}
\end{proof}

\begin{remark}
\label{rmk:analysis-on-k-means}
We now describe the changes needed to generalize \cref{alg:k-means} from the $k$-medians objective to $k$-means. Specifically, since the cost measure in $k$-means is $\sum_{x \in \cX} d^2(x,\cC(x))$ instead of $\sum_{x \in \cX} d(x,\cC(x))$, it will suffice to change our ``distance'' measure to $d^2$. This can be dome by making the modification to the sampling probability of $x_{i}$: instead of using $\min\{1, Q(x_{i}, S)/f\}$, we will use the sampling probability of $\min\{1, Q^{2}(x_{i}, S)/f\}$. The only downstream change that occurs is that we no longer can apply triangle inequality, since $d^2(\cdot,\cdot)$ is no longer a metric. However, we can always employ the approximate triangle inequality of $d^{2}(x,z)\leq 2\cdot (d^{2}(x,y)+d^{2}(y,z))$ (see \Cref{prop:coreset-approx-original} for the usage). Note that the triangle inequality is only used in \Cref{clm:k-means-cost-heavy} and \Cref{prop:coreset-approx-original}, and the rest of the algorithm and the analysis proceeds exactly as in the $k$-median case (\Cref{lem:k-means-S-sample} is affected by the high probability event in \Cref{clm:k-means-cost-heavy}, but the analysis itself does not use triangle inequality). Substituting approximate triangle inequality for the triangle inequality induces an additional constant factor into the objective, which does not effect our overall $O(1)$-approximation. 
\end{remark}

\section{Metric Minimum Spanning Tree}
\label{sec:alg-mst}
In \Cref{sec:alg-k-center} and \ref{sec:alg-k-means} we showed that $\tilde{\Theta}(k)$ strong oracle queries are necessary and sufficient for $k$-clustering tasks. In light of this, a natural question is whether the strong oracle is necessary for all geometric optimization problems. In this section, we demonstrate that, surprisingly, this is not the case for the classic metric minimum spanning tree (MST) problem, so long as the weak distances $\td:\cX^2 \to \R$ are a metric. We refer to this as the \textit{metric-weak oracle model}.   Formally, we show:

\begin{theorem}
\label{thm:mst-alg-metric-case}
There is an algorithm that, given only access to the corrupted metric $\td$ produced by a metric weak oracle (i.e., $(\cX,\td)$ is a metric), with corruption probability $\delta$ such that $1-\delta \geq c$ for some constant $c>0$,  in $O(n^2)$ time produces a tree $\hat{T}$ such that $\mathbb{E}[w(\hat{T})] \leq O(\sqrt{ \log n}) \cdot \min_T w(T)$.
\end{theorem}

Our algorithm itself is simple: it computes the optimal MST $\Tstar_{\td}$ of the corrupted metric $(\cX,\td)$, transforms $\Tstar_{\td}$ into a bounded degree tree $\hat{T}$ (\Cref{alg:modify_tree}) such that $\tw(\hat{T}) \leq 2 \tw(\Tstar_{\td})$, and then outputs $\hat{T}$. 
The analysis of this algorithm, however, is fairly involved. To prove that $\hat{T}$ is a good approximation, it suffices to show that with good probability we have both 

\begin{enumerate}
    \item $w(\hat{T}) \leq  \tw(\hat{T}) + O( \sqrt{\log n}) w(\Tstar_{d})$, where $\Tstar_{d} $ is the optimal MST for $(\cX,d)$. 
    \item $\tw(\Tstar_{d}) \leq O(w(\Tstar_{d}))$.
\end{enumerate}

Then, using $\tw(\hat{T}) \leq 2 \tw(\Tstar_{\td}) \leq 2 \tw(\Tstar_{d})$, we can obtain

\[w(\hat{T}) \leq  \tw(\hat{T}) + O(\sqrt{\log n}) w(\Tstar_{d}) \leq  2\tw(\Tstar_{d})+ O(\sqrt{\log n}) w(\Tstar_{d})  \leq O(\sqrt{\log n}) w(\Tstar_{d}).\] 

We now give the follow description and analysis for our algorithm. To begin with, we first show the bounded-degree tree transformation as prescribed by \Cref{alg:modify_tree}. 

\begin{algorithm}
	\caption{\label{alg:modify_tree} Bounded-Degree Tree Transformation}
	\KwData{Rooted tree $T$ over a metric space $(X,d)$. }
	\KwResult{Rooted tree $\hat{T}$ with degree at most $5$ with $w_d(\hat{T}) \leq 2 w_d(T)$.}
	\textbf{Init:}  $\hat{T} = (V,\hat{E})$, where $\hat{E} = \{\}$ is an empty  edge set. \\
	
	\For{every $u \in T$}{
		Let $C_u $ be the set of children of $u$ in $T$,  and set $k = |C_u|$. \\
		\If{$k \leq 2$}{ Add directed edge $(u,v)$ to $\hat{E}$ for every $v \in C_u$ \\}
		\Else{
		Order the children $C_u = \{x_1,\dots,x_k\}$ so that $d(u,x_1)\leq d(u,x_2 )\leq \dots \leq d(u,x_k)$, and define $x_0 = u$. \\
		For each $i=1,2,\dots,k$, add the directed edge $(x_{\varphi(i)}, x_i)$ to $\hat{E}$. 
		}
		
		Remove all covered points form $V$. \\    
	}
	
\end{algorithm}

The notation of $\varphi(i)$ in \Cref{alg:modify_tree} is defined as the parent index in the \emph{binary complete tree}. Concretely, for any integer $k$, we define $H_k$ to be the unique complete binary labeled tree on $k+1$ vertices such that the level-order traversal of $H_k$ is $0,1,2,\dots,k$. In other words, $H_k$ is a complete binary tree where the zero-th level contains just the root, labeled $0$, the first level contains vertices labeled $1,2$ (left to right ordered), and second the labels $3,4,5,6$, and so on.  Define the mapping $\varphi: \mathbb{Z}_{\geq 1} \to \mathbb{Z}_{\geq 0}$ by $\varphi(i) = j$ where $j$ is the label of the parent of $i$ in the graph $H_k$ (for any $k \geq i$): namely, $\varphi(i) = \lceil i/2\rceil-1$. 

We show that the transformation in \Cref{alg:modify_tree} is always possible in \Cref{prop:modifyTree}.

\begin{proposition}\label{prop:modifyTree}
	Fix any $n$-point metric space $(\cX,d)$ and any spanning tree $T$ of $\cX$. Then the spanning tree $\hat{T}$ produced by \Cref{alg:modify_tree} has degree at most $5$, and satisfies 
 $w(\hat{T}) \leq 2 w(T)$.
	Moreover, the algorithm can be run in $O(n)$ time. 
\end{proposition}
\begin{proof} The runtime of the algorithm is straightforward, so we analyze the other two claims. 
	For any vertex $u$, let $\pi(u)$ be its parent in $T$.
	We first observe that the tree $\hat{T}$ produced has degree at most $5$.  To see this, fix any node $u$, and note that edges are added adjacent to $u$ only when the for loop is called on $u$ and on the parent $\pi(u)$. In the case it is called on $u$, at most two children (out-edges) are added to $u$., and when called on the parent, again at most two children and one parent (in-edge)  are added to $u$.  
	Interpreting the in-edge as a parent and out-edges as children, we have that $\hat{T}$ is a rooted tree with the same root as $T$, where each node has at most $4$ children. We can then define $\hat{\pi}(u)$ to be the parent of $u$ in $\hat{T}$. 
	
	Next, we analyze the cost of the tree. Observe that $w_d(T) = \sum_{u \in X} d(u, \pi(u))$. Thus it suffices to show that 
	\begin{equation}\label{eqn:MST}
			  d(u,\hat{\pi}(u))\leq 2d(u, \pi(u))
	\end{equation}
for any non-root node $u$.  To see this, note that the parent $ \hat{\pi}(u)$ is set when the for loop is called on the parent $v=\pi(u)$ of $u$. In this case, we order the children  $C_v = \{x_1,\dots,x_k\}$ so that  $d(v,x_1)\leq d(v,x_2 )\leq \dots \leq d(v,x_k)$, where $ u = x_i$ for some $i \in \{1,\dots,k\}$. First note that if $k\leq 2$, then we have $\hat{\pi}(u) = v = \pi(u)$, so \eqref{eqn:MST} holds trivially. Otherwise, we have $\hat{\pi}(u) = x_j$ for some $j < i$ (interpreting $x_0 = v$). By the ordering, and employing the triangle inequality, we have 
\begin{equation}
\begin{split}
	d(u,\hat{\pi}(u) )  &= d(x_i,x_j) \\
	& \leq d(x_i,v) + d(x_j,v) \\
	&\leq 2 d(x_i,v) \\
	&= 2d(u, \pi(u)),
\end{split}
\end{equation}
	which completes the proof.	
\end{proof}



Before proceeding with the analysis of our algorithm, for simplicity, we assume that all distances are unique, and we justify this assumption as follows.

\begin{fact}\label{fact:uniquedist}
    Let $\cA$ be any algorithm that satisfies the correctness guarantees of \Cref{thm:mst-alg-metric-case} under the assumption that the set of distance $\{d(x,y)\}_{(x,y) \in \binom{\cX}{2}}$ are unique. Then there is an algorithm $\cA'$ that satisfies the correctness guarantees of \Cref{thm:mst-alg-metric-case} without this assumption. 
\end{fact}
\begin{proof}
    The proof is by modifying the input to $\cA$ to satisfy this. Specifically, we for every $(x,y)$, we replace the input $\td(x,y)$ with $\td(x,y) + \eps_{x,y}$, where $\eps_{x,y} \sim [\eps/2, \eps]$ is i.i.d. and uniform for an arbitrarily small value of $\epsilon$. It is easy to verify that the result is still a metric, as we always have $r_{x,y} \leq \eps \leq r_{x,z} + r_{z,y}$ for any $(x,y,z)$. We run $\cA$ on the modified distances $\tilde{d}(x,y) + r_{x,y}$, which we can interpret as coming from the modified original metric with distances $d(x,y) + r_{x,y}$. Since $\{d(x,y)\}_{(x,y) \in \binom{\cX}{2}}$ are now unique, the correctness of $\cA$ follows. Moreover, since we have changed each distance in $d(x,y)$ by at most $\eps << 1/\poly(n) \cdot \min_{x,y} d(x,y)$, it follows that the cost of any spanning tree is changed by at most a $(1-1/\poly(n))$ factor, which completes the proof. 
\end{proof}

Given \Cref{fact:uniquedist}, in what follows we will always assume that all distances are unique. We can now introduce the notion of a $\ell$-heavy ball, whose definition relies on this fact.

\begin{definition}[$\ell$-heavy ball]
\label{def:level-ell-ball}
Fix any $\ell \geq 0$. 
For any point $v\in \cX$, we define the level-$\ell$ heavy radius at the point $v$ to be the smallest radius $r=r^{\ell}_{v}$ such that the metric ball $\ldball{d}{\ell}{v}{r}$ under the distance measure $d$ contains exactly $2^{\ell}$ points in $\cX$.
We define the level-$\ell$ heavy ball at $v$ to be the metric-ball $\ldball{d}{\ell}{v}{r^\ell_v}$ under the metric $d$. 
\end{definition}
Note that the existence of a radius $r^\ell_v$ such that $\ldball{d}{\ell}{v}{r}$  contains exactly $2^{\ell}$ points is guaranteed by the uniqueness of distances.
We use the notion of $\ell$-th level heavy balls in \Cref{def:level-ell-ball} to show the following probabilistic guarantee for distance corruptions under the metric constraints.

\begin{lemma}[Probabilistic metric violation guarantee]
	\label{lem:ball-distance-preserve}
	Let $\dtilde$ be the corrupted distance of $d$ that satisfies the metric property. Fix any three points  $x,y,u\in V$ and level $\ell \geq 2$, 
	such that $x\in \ldball{d}{\ell}{u}{r^{\ell}_{u}}$.
	Then with probability at least $1-2^{- c_\delta\cdot2^\ell} $, where $c_\delta$ is a constant depending only on the corruption probability $\delta$, the following inequalities hold:
	\begin{align*}
		& \dtilde(x,y) \leq d(x,y) + 4\cdot r^{\ell}_{u}; \\
		& d(x,y) \leq \dtilde(x,y) + 4\cdot r^{\ell}_{u}.
	\end{align*}
\end{lemma}
\begin{proof}
 For any point $z \in \lball{\ell}{u}{r^{\ell}_{u}} \setminus \{x,y\}$, let $\cE_{z}$ be the event that both $\dtilde(y, z) = d(y,z)$ and $\dtilde(x, z) = d(x,z)$; i.e., neither pair is corrupted. Then the probability that no such $\cE_z$ holds for any $z$ is at most
	\begin{align*}
		\Pr\left(\bigcap_{z \in \lball{\ell}{u}{r^{\ell}_{u}} \setminus \{x,y\}}  \neg \cE_z  \right) \leq (1-(1-\delta)^2)^{ 2^{\ell-1}}\leq 2^{- c_\delta \cdot 2^\ell},
	\end{align*}
for some $c_\delta$ which is a constant depending only on $\delta$. Thus, with probability at least $1-2^{- c_\delta \cdot 2^\ell} $, there exists at least one such $z \in \lball{\ell}{u}{r^{\ell}_{u}} \setminus \{x,y\}$ with $\dtilde(y, z) = d(y,z)$ and $\dtilde(x, z) = d(x,z)$. Condition on this event, and fix it $z$ now.  
	We have
	\begin{align*}
		\dtilde(x,y) &\leq \dtilde(x,z) + \dtilde(z,y) \tag{triangle inequaliy for $\dtilde$}\\
		&=d(x,z)+d(z,y) \tag{the event $\cE_z$}\\
		&\leq  2 d(x,z) + d(x,y) \tag{triangle inequality for $d$}\\
		&\leq d(x,y) + 4 \cdot r^{\ell}_{u}. \tag{$x,z \in \lball{\ell}{u}{r^{\ell}_{u}}$}
	\end{align*}
 
	Similarly, for the second inequality, we have
	\begin{align*}
		d(x,y) &\leq d(x,z) + d(z,y) \tag{triangle inequaliy for $d$}\\
		&= d(x,z)+\dtilde(z,y) \tag{the event $\cE_z$}\\
		&\leq 2 r^{\ell}_{u} + \dtilde(x,z) + \dtilde(x,y) \tag{$d(x,z)\leq r_{v}$ and triangle inequality for $\dtilde$}\\
		&= 2 r^{\ell}_{u} + d(x,z) + \dtilde(x,y) \tag{the event $\cE_z$}\\
		&\leq \dtilde(x,y) + 4 \cdot r^{\ell}_{u}, \tag{$d(x,z)\leq r_{v}$ }
	\end{align*}
	as desired.
\end{proof}

We now use \Cref{lem:ball-distance-preserve} to prove the approximation guarantee for the MST. We first define a partition of the metric space $\cX$ via a ball-carving. Note that the following procedure is not algorithmic, and is only used in the analysis. Also, in what follows, recall that we define $\Tstar_{d}$ and $\Tstar_{\dtilde}$ be the MST under the metric $d$ and $\dtilde$ respectively. 

\begin{algorithm}[!htbp]
	\caption{\label{alg:thought-process-mst} Level-$\ell$ Heavy Ball Carving }
\KwData{Set of points $V$, integer $\ell \geq 1$}
\textbf{Output:} Set of metric balls $B_1,\dots,B_k$ covering $V$, and partition $S_1,\dots,S_k$ of $V$ such that $S_i \subseteq B_i$ for all $i \in [k]$. 
\\ \vspace{1 em}
Initialize $i=1$ and $X=V$. \\
\While{$X\neq \emptyset$}{
	$x_i = \argmax_{y \in X} r_y^\ell$. \\
	Set  $r^{\ell}_{i}=r^{\ell}_{v}$,  $B_i = \lball{\ell}{x_i}{r_{i}^\ell}$, and $S^{\ell}_{i}=B_i \cap X$. \\
	$X \leftarrow X \setminus S^{\ell}_{i}$\\
	$i \leftarrow i+1$
	
}

\end{algorithm}

\begin{claim}
\label{clm:mst-cost-lb-ball-carve} Fix any $1 \leq \ell \leq \log n$, and 
let $\Tstar_{d}$ be the minimum spanning tree of $G=(V,E)$ under distance $d$. Then we have 
\begin{align*}
	w_{d}(\Tstar_{d}) \geq \frac{1}{2}\cdot \sum_{i} r^{\ell}_{i},
\end{align*}
where $\{r^{\ell}_{i}\}$ are the radii of the level-$\ell$ ball-carving from \Cref{alg:thought-process-mst}.
\end{claim}
\begin{proof}
To avoid redundancy of notation, we drop the superscript of $\ell$ since the proofs on every $\ell$ are the same. We claim that the balls of radius $\ball{x_i}{\frac{r_i}{2}}$ are disjoint (where $x_{i}$ is the selected center of the ball of the $i$-th iteration). To see this, note that for a fixed $i$, if $\ball{x_i}{\frac{r_i}{2}}$ contains a point of $z\in \ball{x_{j}}{r_j}$ for $j>i$, then the ball $\ball{x_i}{r_i}$ should have contained $x_{j}$ (since $r_{j}<r_{i}$), and $x_{j}$ should not have been selected during the ball carving process, which forms a contradiction. Similarly, if $\ball{x_i}{\frac{r_i}{2}}$ contains a point of $z\in \ball{x_{k}}{r_k}$ for $k<i$, then the ball $\ball{x_k}{r_k}$ should have contained $x_{i}$, and $x_{i}$ should not have been selected during the ball carving process. Finally, since the points in each $\ball{x_i}{\frac{r_i}{2}}$ are disjoint, the minimum spanning tree must pay a cost to travel from the boundary to the center of each disjoint ball, which pays at least $\frac{1}{2}\cdot \sum_{i} r_{i}$ in cost.
\end{proof}
Besides lower-bounding the MST cost as in \Cref{clm:mst-cost-lb-ball-carve}, we present two technical steps toward the proof of \Cref{thm:mst-alg-metric-case}.
\begin{lemma}
\label{lem:mst-corrupt-leq-true}
Let $T$ be any spanning tree of the set of points $\cX$. Then we have
\begin{align*}
	\expect{w_{\dtilde}(T)} \leq w_{d}(T) + O(1)\cdot w_{d}(\Tstar_{d}),
\end{align*}
where the expectation is over the randomness over which distances are corrupted. 
\end{lemma}
\begin{proof}
We root the tree $T$ arbitrarily and define the charging scheme as follows. First note, by  Lemma  \Cref{lem:ball-distance-preserve}, for any pair $(x,y)$, the event in Lemma  \Cref{lem:ball-distance-preserve} holds with probability $1-1/\poly(n)$ for at least one value of $\ell$ with $\ell = O(\log n)$. It follows by a union bound over $O(n^2)$ pairs that, with high probability, all pairs $(x,y)$ satisfy the random event in Lemma  \Cref{lem:ball-distance-preserve} for some  $\ell = O(\log n)$. 

\begin{tbox}
\underline{A Charging Scheme.} \\
 For each tree edge $(u,v)\in T$:
	\begin{enumerate}
		\item Suppose w.l.o.g. that $v$ is the child node. Let $\ell$ be the smallest integer such that the ball $\lball{\ell}{x_{i}}{r^{\ell}_{i}}$ satisfies the following properties:
		\begin{itemize}
			\item $v$ is included in $\lball{\ell}{x_{i}}{r^{\ell}_{i}}$; and
			\item $\dtilde(u,v) \leq d(u,v) + 4\cdot r^{\ell}_{i}$. 
		\end{itemize}
	\item Distribute a charge of $4 \cdot r^{\ell}_{i}$ to $\lball{\ell}{x_{i}}{r^{\ell}_{i}}$
	\end{enumerate}
\end{tbox}

It is straightforward that $\sum_{(u,v) \in T} \td(u,v) \leq \sum_{(u,v) \in T} \td(u,v) + C$, where $C$ is the sum of all charges distributed to balls in the above process. Thus it suffices to upper bound $C$. To this end, we bound the expected number of times for a given ball at level $\ell$ to be charged. Note that such a ball $\lball{\ell}{x_{i}}{r^{\ell}_{i}}$ can be charged at most once for each of the at most $2^\ell$ points $x$ contained in that ball. Moreover, to be charged by the point $v$, it must be that we did not have $\dtilde(u,v) \leq d(u,v) + 4\cdot r^{\ell-1}_{i}$, which occurs with probaiblity at most $2^{- c_\delta 2^{\ell-1}}$ by  \Cref{lem:ball-distance-preserve}.  Define $X_i^\ell$ to be the number of times that $\lball{\ell}{x_{i}}{r^{\ell}_{i}}$ is charged. Then, letting $X^\ell(p)$ be the event that $\ell$ is the level at which the point $p$ is charged in the above scheme, we have
\begin{align*}
	\expect{X_i^\ell} &\leq \sum_{v\in \lball{\ell}{x_{i}}{r^{\ell}_{i}}} \expect{X^\ell(v)}\\
 &\leq\sum_{v\in \lball{\ell}{x_{i}}{r^{\ell}_{i}}} 2^{- c_\delta 2^{\ell-1}}\\
	&\leq2^\ell 2^{- c_\delta 2^{\ell-1}}\leq \frac{c_\delta'}{2^{\ell}}
\end{align*}
for some $c_\delta' \geq 0 $ which is another constant.  
Thus, we can bound the total cost of the charging scheme by 
\begin{align*}
\sum_{\ell=1}^{\log{n}} \sum_{i}	\expect{ 4 r_i^\ell X_i^\ell}\leq \sum_{\ell=1}^{\log{n}} \sum_{i} 4  \frac{c_\delta'}{2^{\ell}}  \cdot r^{\ell}_{i} \leq \sum_{\ell=1}^{\log{n}} \frac{8}{2^{\ell}} c_\delta' w_d(T^*_d) = 16 c_\delta' w_d(T^*_d),
\end{align*}
where the last inequality follows from \Cref{clm:mst-cost-lb-ball-carve} and the fact that the sum is geometric. 
Putting these bounds together, we have
\begin{align*}
	\expect{w_{\dtilde}(T)} &\leq \sum_{(u,v)\in T}d(u,v) +\sum_{\ell=1}^{\log{n}} \sum_{i}	\expect{ 4 r_i^\ell X_i^\ell}\\
		&\leq w_{d}(T) + O(1) \cdot w_{d}(\Tstar_{d}).
\end{align*}
as desired.
\end{proof}


\paragraph{Lower Bounding $\tw(T)$ for any tree $T$}
We now prove an inequality in the reverse direction, demonstrating that, with high probability over the choice of corrupted distances, for any spanning tree $T$ the cost of $\tw(T)$ is not too small.  We begin with the following definition. In what follows, we fix $\alpha = \Theta(\sqrt{\log n})$ with a sufficiently large constant.  For any $\ell$, we will write $B_1^\ell,B_2^\ell,\dots,B_k^\ell$ to denote the set of balls produced by the Level-$\ell$ Heavy Ball Carving ( \Cref{alg:thought-process-mst}), and let $r_i^\ell$ denote the radius of $B_i^\ell$. Note that by construction, each ball $B_i^\ell$ contains exactly $2^\ell$ points.
In what follows, set $\beta = 1-\delta$, and note that $\beta = \Omega(1)$ is at least a constant. 

\begin{definition}
	Fix $\alpha$ as above, and fix any $x \in \cX$. Let $B_i^\ell$ be the ball in the level-$\ell$ heavy ball carving containing $x$, where $\ell = \log (\alpha)$. Then we say that $x$ is \textit{good} if at least $\beta 2^{\ell-1}$ distances in the set $\{(x,y) \; | \; y \in B_i^\ell\}$ are not corrupted. 	Call a point \textit{bad} if it is not good. 
\end{definition}

\begin{proposition}\label{prop:MST1}
	With probability $1-1/\poly(n)$, the following holds: for every pair of two good points $(x,y)$,  we have 
	\[ d(x,y) \leq \td(x,y)   + 4(r_i^\ell + r_j^\ell)\]
where $(i,j)$ is such that $x \in B_i^\ell,y \in B_j^\ell$, and $\ell = \log \alpha$.
\end{proposition}
\begin{proof}
	We prove that for any two balls $B_i^\ell, B_j^\ell$ (possibly with $i=j$), and any subsets $S_i \subset B_i^\ell, S_j \subset B_j^\ell$ with size $|S_i|, |S_j| \geq \beta 2^{\ell -1}$, there exists at least one $w \in S_i, z \in S_j$ such that $(w,z) \notin \cor$. Let $\cE$ denote this event, we prove that $\Prob[\cE] \geq 1- 1/\poly(n)$. Given any two sets $S_i,S_j$, there are at least $\binom{\beta 2^{\ell-1}}{2} =s = \Theta(\alpha^2)> 100  \log n/\beta$ distinct pairs $(w,z) \in S_i \times S_j$ (where we used that $\alpha$ is set with a sufficiently large constant depending on $1/\beta$). Since each pair is corrupted independently with probability $\delta$, the probability that all pairs $(w,z) \in S_i \times S_j$ are contained in $\cor$ is at most
 \[ (\delta)^{\frac{100 \log n}{(1-\delta)}} =        (1-(1-\delta))^{\frac{100 \log n}{(1-\delta)}} \leq \left(\frac{1}{2}\right)^{100 \log n} = n^{-100},\]
 where the inequality follows using the fact that $(1-x)^r \leq \frac{1}{1+rx}$ for all $x \in [-\frac{1}{r},1]$ and $r \geq 0$. We can then union bound over all $\binom{n}{2} 2^\alpha = O(n^2) \cdot  2^{O(\sqrt{log n} )} < O(n^3)$ such choices of $S_i,S_j$ to obtain the desired result with probability at least $1-n^{-95}$.

	Now conditioned on $\mathcal{E}$, we can fix any two points $x,y$ as in the statement of the proposition, where  $x \in B_i^\ell,y \in B_j^\ell$. Since $x,y$ are both good, there exist the desired sets $S_i \subset B_i^\ell, S_j \subset B_j^\ell$ with size $|S_i|, |S_j| \geq \beta 2^{\ell -1}$, such that all pairs $(x,u) \in S_i$ and $(y,v) \in S_j$ are not corrupted. By $\mathcal{E}$, we can then fix $w \in S_i$ and $z \in S_j$ such that $(w,z) \notin \cor$ is also not corrupted. We then have	
	\begin{equation}
		\begin{split}
			d(x,y)& \leq d(x,w) + d(w,z) + d(z,y) \\
			& \leq  2r_i^\ell + \td(w,z) + 2r_j^\ell \\ 
			& \leq 2(r_i^\ell + r_j^\ell) + \td(w,x) + \td(x,y) + \td(y,z) \\ 
			&= 2(r_i^\ell + r_j^\ell)+ d(w,x) + \td(x,y) + d(y,z) \\ 
				&= 4(r_i^\ell + r_j^\ell)+ \td(x,y),  \\ 
		\end{split}
	\end{equation}
which completes the proof. 
\end{proof}

We now consider how to bound the cost of edges $(x,y)$ where at least one of $x,y$ is bad. To do so, set $\ell^*$ such that $2^{\ell^*} =  c^*\log n$ with a large enough constant $c^*$, and consider the level-$\ell^*$ heavy ball carving $B_1^{\ell^*},B_2^{\ell^*},\dots,B_t^{\ell^*}$. Recall that each such ball has exactly $c^* \log n$ points. We have  the following.
\begin{fact}\label{fact:MST1}
	With probability at least $1-1/\poly(n)$, for every $i \in [t]$, the ball $B_i^{\ell^*}$ contains at most $\sqrt{\log n}$ bad points. 
\end{fact}
\begin{proof}
Note that for any point $x \in B_i^{\ell^*}$ , we have 
$$\expect{\left|\{(x,y) \; |\; y \in B_i^{\ell^*}\} \cap \cor\right|} = \delta 2^{\ell^*}$$ 
And recall that $x$ is bad if $\left|\{(x,y) \; |\; y \in B_i^{\ell^*}\} \cap \cor\right| < \delta 2^{\ell^* - 1}$. 
Thus, by Chernoff bounds, a point is bad with probability at most $2^{-\Theta(\alpha)} < 2^{-100 \sqrt{\log n}}$. Thus, the probability that any fixed set $S$ of $\sqrt{\log n}$ points is simultaniously bad is at most $ (2^{-100 \sqrt{\log n} })^{\sqrt{\log n}} = 1/n^{100}$. It follows that the probability that any set of more than $\sqrt{\log n}$ points in $B_i^{\ell^*}$ is bad is at most 
	\begin{equation}
		\begin{split}
			\Pr\left[ B_i^{\ell^*} \text{ contains at least }  \sqrt{\log n} \text{ bad points }\right]   			&\leq  \binom{c^* \log n}{\sqrt{\log n}} \cdot n^{-100}\\
			&\leq  (c^* \log n)^{\sqrt{\log n}} n^{-100}\\
			&\leq n^{-99} \\
		\end{split}
	\end{equation}
Union bounding over $t \leq n$ possible balls yields the desired result.
\end{proof}

\begin{fact}\label{fact:MST2}
	With probability at least $1-1/\poly(n)$, the following holds: for every pair $(x,y) \in \cX$, where $x \in B_i^{\ell^*}$, we have 
	\[ d(x,y) \leq \td(x,y) + 4 r_i^{\ell^*}.   \]
\end{fact}
\begin{proof}
	Note that for any $x$, with probability at least $(1-(1-\delta))^{c^* \log n} > 1-n^{-100}$, there exists at least one $z \in B_i^{\ell^*}$ with $(x,z) \notin \cor$ and $(z ,y )\notin \cor$. Conditioned on this, we have:
	
	\begin{equation}
		\begin{split}
			d(x,y) & \leq d(x,z) + d(z,y) \\
			& \leq 2 r_i^{\ell^*} + \td(z,y) \\
				& \leq 2 r_i^{\ell^*} + \td(z,x) + \td(x,y) \\
					& \leq 4r_i^{\ell^*} +  \td(x,y). \\
		\end{split}
	\end{equation}
The fact follows after union bounding over $O(n^2)$ pairs $(x,y)$. 
\end{proof}

\begin{proposition}\label{prop:MSTFinal}
	With probability at least $1-1/\poly(n)$, the following holds:  for every spanning tree $T$ of $\cX$ with degree at most $\Delta$, we have
	\[ w(T) \leq  \tw(T) + O(\Delta \sqrt{\log n}) \min_{T'} w(T').  \]
\end{proposition}
\begin{proof}
	We first condition on the events in \Cref{prop:MST1}, and \Cref{fact:MST1} and \Cref{fact:MST2}, which all occur with probability $1-1/\poly(n)$. For any edge $(x,y) \in T$, if both $x,y$ are good, we have
	\begin{equation}\label{eqn:MST3}
		d(x,y) \leq \td(x,y) + 4(r_i^{\ell} + r_j^\ell),
	\end{equation}	
	where $\ell = \log(\alpha)$, and $x \in B_i^\ell, y\ \in B_j^\ell$. Otherwise, if at least one of $x,y$ is bad, then fix one of the points which is bad, w.l.o.g.\ we fix $x$ which is bad. Then by \Cref{fact:MST2}, we have 
		\begin{equation}\label{eqn:MST4}
			d(x,y) \leq \td(x,y) + 4 r_\tau^{\ell^*}),
				\end{equation}			
	where $x \in B^{\ell^*}_\tau$. Now to bound the cost $\sum_{(x,y) \in T} d(x,y)$, we will bound each summation by either \Cref{eqn:MST3} or \Cref{eqn:MST4}, depending on whether both $x,y$ are good or if at least one is bad. We now bound the total cost of doing so. 	
	
	Using \Cref{fact:MST1}, we know that each ball $B^{\ell^*}_\tau$ has at most $\sqrt{\log n}$ bad points points. Moreover, this ball can only contribute a cost of $4 r_\tau^{\ell^*}$ at most $O(\Delta)$ times for each bad point in $B^{\ell^*}_\tau$. Thus, over all edges $(x,y) \in T$, the term $4 r_\tau^{\ell^*}$ appears on the RHS of the above equation at most $O(\Delta \sqrt{\log n})$ times. Similarly, for ball $B^{\ell}_j$ at level $\ell = \log \alpha$, the term $r_i^{\ell}$ can only appear in the RHS of \Cref{eqn:MST3} when considering an edge with at least one endpoint in $B^{\ell}_j$. Since $|B^\ell_j| < O(\sqrt{\log n})$, again this radius is counted at most $O(\Delta \sqrt{\log n} )$ times. It follows that	
	\begin{equation}
		\begin{split}
			\sum_{(x,y) \in T} d(x,y) &\leq \sum_{(x,y) \in T} \td(x,y) + O(\Delta \sqrt{\log n})\left(\sum_i r^\ell_i + \sum_j r^{\ell^*}_j\right) \\
			&\leq \sum_{(x,y) \in T} \td(x,y) + O(\Delta \sqrt{\log n}) \min_{T'} w(T') ,
		\end{split}
	\end{equation} 
as needed, where we used \Cref{clm:mst-cost-lb-ball-carve} in the last inequality. 
\end{proof}

We are now ready to prove \Cref{thm:mst-alg-metric-case}.
\begin{proof}[Proof of \Cref{thm:mst-alg-metric-case}]
Let $T^* = \arg \min_T w(T)$. Letting $\tT = \arg \min_T \tw(T)$, we then set the output of our algorithm to be the result  $\hat{T}$ of running \Cref{alg:modify_tree} on $\tT$ in the corrupted metric $\td$. By \Cref{prop:modifyTree}, the tree $\hat{T}$ has degree at most $5$ and $\tw(\hat{T}) \leq 2 \tw(\tT)$. Then by \Cref{prop:MSTFinal}, we have
\begin{equation}
	\begin{split}
		\ex{w(\hat{T})} \leq & \ex{\tw(\hat{T})} + O(\sqrt{\log n}) w(T^*) \\
        \leq & 2 \ex{\tw(\tT)}+ O(\sqrt{\log n}) w(T^*) \\ 
		\leq & 2 \ex{\tw(T^*)}+ O(\sqrt{\log n}) w(T^*) \\ 
		\leq & 2 w(T^*) + O(1) \cdot w(T^*) + O(\sqrt{\log n}) w(T^*) \\ 
		= & O(\sqrt{\log n}) w(T^*), \\ 
	\end{split}
\end{equation}
where in the first line we applied \Cref{prop:MSTFinal}, the second line used that $\hat{T}$ was a $2$-approximation of the optimal MST in the corrupted space $\td$, the third line used that $\tT$ is optimal for $\td$, and in the fourth line we applied \Cref{lem:mst-corrupt-leq-true}.
\end{proof}

\section{Lower Bounds}\label{sec:LB}
We give lower bounds for $k$-clustering and MST in this section. In particular, we show that
\begin{itemize}
\item For $k$-clustering, we show that any $k$-center, $k$-means, or $k$-median algorithm that provides \emph{bounded} approximation under the weak-strong oracle model requires $\Omega(k)$ strong (point) oracle queries.
\item For metric MST, we show that if we want to go below the approximation factor of $\sqrt{\log n}$, we have to make $\tilde{\Omega}(n)$ strong (point) oracle queries.
\item Finally, for non-metric MST, we show that with $o(n)$ strong (point) oracle queries, we cannot break an approximation lower bound of $\Omega(\log{n})$.
\end{itemize}
Our lower bound demonstrates that the algorithms we designed in \Cref{sec:alg-k-center,sec:alg-k-means,sec:alg-mst} are nearly tight up to $\polylog{n}$ factors, and one could not hope for algorithms that are significantly more efficient. Furthermore, by our lower bound on non-metric MST, we separate the complexity between the metric and non-metric cases.

\subsection{Lower Bounds for $k$-Clustering}
In the prior sections, we provided $O(1)$-approximation algorithms for $k$-center, $k$-means and $k$-median clustering, each using $\tilde{O}(k)$ queries to the strong oracle. A natural question is whether the strong location oracle is even necessary for these tasks: namely, is it possible to obtain a good approximation with the weak oracle alone? We demonstrate that this is not possible in a strong sense. Namely, we prove that $\Omega(k)$-strong oracle queries are necessary for any algorithm that achieves \textit{any} bounded approximation for $k$-clustering tasks. Our main result is as follows:

\begin{theorem}
\label{thm:k-clustering-lb}
Fix any positive real number $c\in \mathbb{R}^{+}$, and positive integer $k$ larger than some constant, and fix the corruption probability to be $\delta = 1/3$. Then any algorithm $\ALG$ which produces a solution for either $k$-centers, $k$-means, or $k$-medians that, with probability at least $1/2$, has cost at most $c \cdot \OPT$, (where $\OPT(d)$ is the optimal solution to the clustering task in question) must make at least $\Omega(k)$ queries to a strong (point) oracle, or at least $\Omega(k^2)$ queries to a strong distance oracle.%
\end{theorem}
\begin{proof}
 We focus on the proof of the $k$-center case, and the lower bounds for $k$-means and $k$-median clustering follow the same construction. The construction of the hard distribution over inputs is as follows. The distribution will be over distances $d$ for a fixed set of $n$ points $\cX $. We first assume that $k$ is odd, and later generalize to the case of even $k$. Moreover, since any $s$-query point strong oracle  algorithm implies a $s^2$-query edge strong oracle algorithm, it suffices to prove a $\Omega(k^2)$-query lower bound against edge strong oracles, since this will imply a $\Omega(k)$ lower bound for point strong oracle queries.
 
\begin{mdframed}
\textbf{Construction of the ground-truth metric $(\cX,d)$.}
\smallskip
\begin{enumerate}
\item Partition $\cX$ into sets $S,O$, so that $S$ has the first $|S| = \frac{3}{2}(k-1)$ points (under some fixed ordering), and $O$ has all remaining points. 
\item Select a uniformly random subset $N \subset S$ of exactly $k-1$ points, and define $U = S \setminus N$. 
\item Fix a uniformly random perfect matching $M$ over $N$, so that $|M| = \frac{k-1}{2}$. 
\item Define the metric $d$ as follows:

\[     d(x,y) = \begin{cases}
    1 & \text{ if } (x,y) \in M \\ 
    1 & \text { if } (x,y) \in O \times O \\
    c & \text {otherwise }\\
\end{cases}\]

\end{enumerate}
\end{mdframed}
It is straightforward to verify that the construction of $d$ is a metric. We now describe how to generate the corrupted distances $\tilde{d}$ for a given draw of $d$. Specifically, the weak oracle will corrupt \textit{at most one} distance $d(x,y)$. 

Observe that the optimal $k$-centers clustering of the original metric $(\cX,d)$ has a cost of $1$, and must have exactly one center in $O$, one center chosen from each of the matched pairs $(x,y) \in M$, and one center placed at every unmatched point $y \in U$. Note that if a single one of these clusters does not have a center placed in it, the cost of the solution is at least $c$.

\begin{mdframed}
\textbf{Construction of the weak-oracle metric $\tilde{d}$}
\begin{enumerate}
\item\label{line:adversary-k-cluster} Fix an arbitrary pair $(x^*,y^*) \in M$ such that $(x^*,y^*) \in \cor$ is corrupted. If no such pair exists, set $\td = d$, otherwise:
\item Set $\td(x^*,y^*) = c$, and for all other pairs $(x,y) \in \binom{n}{2} \setminus \{(x^*,y^*)\}$, set $\td(x,y) = d(x,y)$. 
\end{enumerate}

\end{mdframed}

Let $\mathcal{E}_1$ be the event that at least one pair $(x,y) \in M$ exists such that $(x,y) \in \cor$. Note that $\prob{\cE_1} > 1-(2/3)^{(k-1)/2} = 1- 2^{-\frac{k-1}{4}}$. We now condition on this holding. Now consider the metric $\td$ produced by the weak oracle conditioned on $\cE_1$. Now consider the distances $\td$ produced by the weak oracle. They consists of the cluster $O$ of points pairwise distance $1$ apart within $O$, and distance $c$ away from all points not in $O$. It also consists of the matching $M' = M \setminus (x^*,y^*)$, where $\td(x,y) = 1$ for all $(x,y) \in M'$, and then it consists of the $k/2$  points $S \cup \{x^*,y^*\}$ which are each distance $c$ from all other points in $\cX$. Notice, however, that the pair $x^*,y^*$ that was corrupted was not known to the algorithm. Moreover, since $N$ was chosen uniformly at random, and the matching $M$ was uniformly random, if we let $T$ be the set of $|T| = k/2$ points in $\td$ that are distance $c$ from all other points, it follows that the identity of the corrupted distance $(x^*,y^*)$ is a uniformly random pair chosen from $T$. 

Now consider any sequence $s_1,s_2,\dots, \in \binom{\cX}{2}$ of adaptive, possible randomized edge strong oracle queries made by an algorithm. Since $d(x,y) = \td(x,y)$ for all pairs $x,y$ such that at least one of $x,y \notin T$, we can assume WLOG that each $s_i \in \binom{T}{2}$ (otherwise it reveals no information to the algorithm). Now for any prefix $s_1,\dots,s_i$, condition on the event $\cQ_i$ that $s_1 \neq (x^*,y^*), s_2 \neq (x^*,y^*), \dots s_i \neq (x^*,y^*)$. Conditioned on $\cQ_i$, it still holds that the single corrupted pair $(x^*,y^*)$ is still uniformly distributed over the set $\binom{T}{2} \setminus \{s_1,\dots,s_i\}$. Thus, for any $s_{i+1}$ with $i+1 < k^2/100$, we have 
\begin{equation}
    \begin{split}
\prob{s_{i+1} \neq (x^*,y^*) \; |\;  \cQ_i } &= \prob{ \cQ_{i+1} \;| \;\cQ_{i} } \\
&= 1- \frac{1}{\binom{|T|}{2} - i} \\
& > 1-\frac{16}{k^2} \\
    \end{split}
\end{equation}

Thus, if the algorithm makes a total of $\ell < k^2 / 1600$ strong edge oracle queries, we have 
\[\prob{\cQ_\ell }  >  \left(1-\frac{16}{k^2}\right)^\ell > 24/25\]
It follows that, conditioned on $\cE_1$, with probability at least $24/25$, the algorithm does not find the corrupted pair. Condition on this event $\cQ_\ell$ now, and condition on any output clustering $\cC$ of the algorithm given the observations $s_1,\dots, s_\ell$. First, suppose that $\cC$ does not contain exactly $k/2-1$ clusters in the set $T$. If it contains more, then either it does not contain a center in $O$, or it does not contain a center in a matching $(x,y) \in M'$, in either case it pays a cost of $c$. Thus, we can assume it has exactly $k/2-1$ clusters in $T$. Let $z \in T$ be the one point in $T$ not opened as a center. If $z \notin \{x^*,y^*\}$, then clear $\ALG$ pays a $k$-centers cost of $c$. We show this happens with good probability. 

To see this, note that since the oracle queried at most $\ell < k^2/1600$ points, it follows that the corrupted distance $(x^*,y^*)$ is still uniformly distributed over the $\binom{k/2}{2}-\ell > k^2/32$ distances not queried within $T \times T$. Since at most $k/2$ of those distances can include $z$, the probability that $z \in  \{x^*,y^*\}$ is at most $\frac{16}{k}$, in which case the algorithm pays a $c$ approximation. Thus, conditioned on $\cE_1,\cQ_\ell$, the algorithm $\ALG$ still pays a $c$ approximation with probaiblity at least $\frac{16}{k}$. Thus, by a union bound, $\ALG$ pays a $c$-approximation with probability at least $1-(\frac{1}{25} + 2^{-\frac{k-1}{4}} + \frac{16}{k}) > 1/2$, which completes the proof for odd $k$.

Lastly, to handle the case when $k$ is even, we can use the same instance, except take a final point $w^*$ from $O$ and make it distance $c^2$ from all other points in both $d$ and $\td$ -- a center must be placed at $w^*$, and the remaining problem is reduced to the above instance with $k-1$ centers (which is now odd). Finally, while the above lower bound was for $k$-centers, note that the same instance implies a $\Omega(c/n)$ approximation lower bound against algorithms with the same query complexity for either $k$-means or $k$-median for algorithms. Since $c$ can be made arbitrarily large, the result for $k$-means and $k$-medians follows. 
\myqed{\Cref{thm:k-clustering-lb}}
\end{proof}


\subsection{Lower Bounds for Minimum Spanning Trees}
In this section, we prove lower bounds for both the metric and non-metric MST problems.

\subsubsection{Lower bound for Metric Minimum Spanning Tree.} We now prove a matching lower bound for the metric MST problem. Our construction is based on a instance with $O(n/\sqrt{\log n})$ well-seperated clusters $\{C_i\}_i$. We show that, with good probability, we can match nearly all clusters into pairs $(C_i,C_j)$ such that \textit{all} distances between $C_i,C_j$ are corrupted. By corrupting these distances it will be impossible to recover the original clusters, which we show implies a $\Omega(\sqrt{\log n})$ approximation.

\begin{theorem}\label{thm:mstLBMain}
  There exists a constant $c$ such that any algorithm which outputs a spanning tree $T$ of $(\cX,d)$ such that $\expect{w(T)}\leq c  \sqrt{\log n}\cdot \min_{T'}w(T') $ in the weak-strong oracle model, must make at least $\Omega(n/\sqrt{\log n})$ queries to the strong oracle. Moreover, this holds even when the weak-oracle distances $\td:\cX^2 \to \R$ is restricted to being a metric, and when the corruption probability is $\delta = 1/3$. 
\end{theorem}

To prove \Cref{thm:mstLBMain}, we use the following standard result on large size matching in random graphs. We also provide a proof for completeness.
\begin{fact}\label{fact:largeMatching}
    Let $G = (V,E)$ be a random graph where each edge $(i,j)$ exists independently with probability at least $\rho > c \log n / n$, for a sufficiently large constant $c$. Then with probability $1-1/\poly(n)$, there exists a matching $M \subset \binom{n}{2}$ in $G$ with size at least $|M|> n/4$. 
\end{fact}
\begin{proof}
    The proof is a simple application of the principle of deferred decisions. Order the vertices arbitrarily $x_1,x_2,,\dots,x_n$. Let $Z_{i,j}$ be an indicator random variable for the event that $(x_i,x_j) \in E$.  We build a set of matched points $M \subset [n]$. Initially, $M$ is empty. We will walk through the points $x_i$ for $i=1,2,\dots,(3/4)n$, and show that each can be matched to a vertex if it was not previously matched already.
    
     We first condition on the event $\cal{E}_1$ that $Z_{1,j}$ exists for at least one $x_j \notin M$, which occurs with high probability by a Chernoff bound. Fix that $x_j$ to match to $x_1$, and add both $x_j,x_1$ to $M$. Now for $i=2,\dots,3n/4$, either $x_i$ is matched by step $i$, or we have that $\sum_{j > i, x_j \notin M} Z_{i,j}$ is a sum of i.i.d. indiacator variables with expectation at least $\rho (3n/4 - M)$. So if $|M| > n/2$, then we are done, otherwise $\expect{\sum_{j > i, x_j \notin M} Z_{i,j}}>(c/4) \log n/n$. Thus, again by Chernoff bounds, with high probability there exists at least one $j > i$ with $x_j \not in M$ such that $(x_i,x_j)$ is an edge, and we can match $(x_i,x_j)$ and continue. Since each vertex $x_i$ is matched with high probability for $i=1,2,\dots,(3/4)n$, the fact follows from a union bound. \myqed{\Cref{fact:largeMatching}}
\end{proof}

We now present the main lower bound of \Cref{thm:mstLBMain}, for which we will employ the following input distribution over $(\cX,d,\td)$. 

\begin{mdframed}
\textbf{The Hard Instance for MST}
\begin{enumerate}
\item Set $k = \Theta(\sqrt{ \log n})$, and draw a uniformly random mapping $f:[n] \to [n/k]$ conditioned on $|f^{-1}(j)| = k$ for all $i \in [n/k]$. Define the $i$-th block $B_i = f^{-1}(i)$, and $B = \{B_1,\dots,B_{n/k}\}$.
\item  Define the true distances as follows: we set $d(x,y) = 1$ for any pair $x,y \in B_i$ that are in the same block $B_i$ for some $i$, and $d(x,y) = k$ otherwise.
\item Find a maximal matching $M \subset B \times B$ such that for all $(i,j) \in M$, and for all $x \in B_i, y \in B_j$ we have $(x,y) \in \cor$.
\item For all $(i,j) \in M$, and for all $x \in B_i, y \in B_j$ set the weak oracle distance to be $\tilde{d}(x,y) = 1$. For all other pairs $(x',y')$, set $\tilde{d}(x',y') = d(x',y')$
\end{enumerate}

\end{mdframed}

\begin{proof}[Proof of \Cref{thm:mstLBMain}]

First, note that it is easy to verify that the resulting corrupted distances $\tilde{d}$ are metric. This can be seen for the following reason: any set of distances $d': \cX^2 \to \R$ defined by a partition $P_1,\dots,P_t$ of $[n]$ such that $d'(x,y) = 1$ for $x,y$ in the same piece $P_i$ of the partition, and $d'(x,y) = \ell$ otherwise, for some $\ell > 1$, is a metric. Finally, we note that that both $d,\tilde{d}$ are of this form.

We first prove the lower bound against algorithms that make \emph{no} strong oracle queries. First, note that for any pair of blocks $B_i,B_j$, there are at most $k^2 < \log(n)/200$ pairs of distances $(x,y) \in B_i \times B_j$. The probability that all such pairs are corrupted is at most $(1/3)^{\log(n)/200} > n^{-1/100}$. Thus, by \Cref{fact:largeMatching}, with probability $1-1/\poly(n)$ the matching $M$ satisfies $|M| > \frac{n}{4k}$. Let $\cE_1$ be the event that the matching is at least this large -- we will now condition on $\cE_1$ holding. 

Now fix any draw of the corrupted distances $\td$ observed by the algorithm. Also condition on the set of corrupted distances $\cor$, and the matching $M$ --- we will prove the lower bound even against an algorithm that is told the matching $M$ over $B \times B$. Note that conditioning on $\td,\cor,M$ does not determine the original metric $d$ --- specifically, the function $f$ is not fully determined by $\td,\cor,M$. Since an algorithm that makes no strong oracle queries sees only $\td,M$, by Yao's min-max principle we can assume the algorithm is deterministic, and thus produces a tree $T$ deterministically based on $\td,M$ which, for the sake of contradiction, we suppose satisfies 
     $\expect{w(T)}\leq c'  \sqrt{\log n}\cdot \min_{T'}w(T') $, where the expectation is taken over the remaining randomness in $d$ after conditioning on $\td,\cor,M$. 

     Now fix any arbitrary rooting of $T$, and let $\pi(u)$ be the parent of any vertex $u \in \cX$ under this rooting. We will charge to each vertex $u \in \cX$ the cost $d(u,\pi(u))$. We now condition on any set of identities of $B_j = f^{-1}(j) \subset [n]$ for every block $B_j$ that is not matched under $M$. Additionally, for every matched pair of blocks $B_i,B_j$, we condition on the set of identities in the union $B_i \cup B_j$, but we \textit{do not} condition on the individual sets $B_i$ and $B_j$. Specifically, note that after conditioning on $B_i \cup B_j$ for any $x \in B_i \cup B_j$, we claim that $\prob{f(x) = i} = \prob{f(x) = j} = 1/2$. This holds because even conditioned on $\td$, we have $\tilde{d}(x,y) = 1$ for ally $x,y \in B_i \cup B_j$, so shuffling the values of the identities in $B_i \cup B_j$ does not effect the observations of the algorithm.

     Now consider any $u \in B_i$ such that $(B_i,B_j) \in M$ is matched. First, suppose that $\pi(u) \notin B_i \cup B_j$ -- then $d(u,\pi(u)) = k$ for all possible realizations of the remaining randomness. If $\pi(u) \in B_i \cup B_j$, we claim that $d(u,\pi(u)) = k$ with probability $1/2$ over the remaining randomness in $f$. To see this, note that because $f_i$ maps each point in $B_i \cup B_j$ to $B_i$ or $B_j$ uniformly at random (subject to the constraint that $|f^{-1}(j)| = |f^{-1}(j)| = k$). The constraint only makes it less likely that any pair $(u,\pi(u))$ are mapped to the same side, so:     
     \[\prob{f(u,\pi(u)) =(i,j)} + \prob{f(u,\pi(u)) =(j,i)} \geq \prob{f(u,\pi(u)) =(i,i)} + \prob{f(u,\pi(u)) =(j,j)}  \]
Moreover, whenever $f(u,\pi(u)) =(i,j)$, we have that $d(u,\pi(u)) = k$, which completes the claim. Given this, it follows that the expected value of $d(u,\pi(u))$ is at least $k/2$ for any $u$ in a matched block $B_i$. Since $|M| > n/(4k)$, it follows that at least $n/2$ points are matched, and each has an edge to its parent in $T$ with expected cost $k/2$, from which it follows that the expected cost of $T$ is $\Omega(nk) = \Omega(n \sqrt{\log n})$. Since the true MST cost of $(\cX,d)$ is always at most $O(n)$, resulting by creating a star on the set of points within each $B_i$, and then adding the edges for an arbitrary spanning tree with $n/k-1$ vertices over the vertices $\{p_1,\dots,p_{n/k}\}$, where $p_i \in B_i$ is an arbitrary representative vertex in $B_i$. This completes the proof of the lower bound against algorithms which make no strong oracle queries.

We now show how to generalize the above argument to algorithms that make at most $\frac{n}{100k}$ strong oracle queries. Similar to the above, we condition on the weak oracle mapping $\td$ as well as the matching $M$. We now consider any set of $\frac{n}{100k}$ strong oracle queries made by the algorithm -- let $S \subset [n]$ be the set of vertices queried. Since $|S| < \frac{n}{100k}$ and $|M| > \frac{n}{4k}$, it follows that there is a matching $M'$ with $M' > \frac{n}{8k}$ such that for every $(B_i,B_j) \in M'$, we have $S \cap (B_i \cup B_j) = \emptyset$. It follows that, even after revealing the values of $d(x,y)$ for all $x,y \in S$, for every $x \in B_i \cup B_j$ where $(B_i,B_j) \in M'$, the function $f(x)$ is still uniformly distributed in $\{i,j\}$. The remainder of the arguement follows as above, with a loss of $2$ in the expected cost of the algorithm attributed to the fact that we only have a matching of size $\frac{n}{8k}$
     rather than $\frac{n}{4k}$. \myqed{\Cref{thm:mstLBMain}}
\end{proof}

\begin{remark}
One may wonder whether we can extend the metric MST lower bound to strong \emph{distance} oracle in the same manner of \Cref{thm:k-clustering-lb}. Alas, with our analysis, we cannot get a lower bound as strong as $\tilde{\Omega}(n^2)$. By a simple black-box reduction, \Cref{thm:mstLBMain} implies a $\Omega(n/\sqrt{\log{n}})$ lower bound for strong distance oracle queries for any algorithm with $o(\sqrt{\log{n}})$ approximation. For estimating the \emph{value} of the MST, this turns out to be (nearly) tight as there exists a $O(1)$ approximation with $\tilde{O}(n)$ queries by \cite{CzumajS04}. Exploring whether this is the case for constructing the actual MST with strong distance queries is an interesting direction to pursue.
\end{remark}

\subsubsection{Lower Bounds for Non-Metric Minimum Spanning Tree.}
We now consider the problem of computing an approximate MST in the general Weak-Strong Oracle model, where the corrupted weak-oracle distances $\td$ is not necessarily a metric (i.e., $\td$ can violate the triangle inequality). Whereas \Cref{thm:mst-alg-metric-case} demonstrates that a $O(\sqrt{\log n})$ approximation is possible in the metric-weak oracle case with \textit{no} strong oracle queries. We now prove a $\Omega(\log n)$ approximation lower bound for any algorithm in the non-metric case, even if it makes $o(n/\log n)$ strong oracle queries, thereby strongly separating the two models. 

\begin{theorem}\label{thm:mstLBGeneralMain}
  There exists a constant $c$ such that any algorithm which outputs a spanning tree $T$ of $(\cX,d)$ such that $\expect{w(T)}\leq c  \log n\cdot \min_{T'}w(T') $ in the weak-strong oracle model (with corruption probability $\delta = 1/3)$, must make at least $\Omega(n)$ queries to the strong oracle. 
\end{theorem}
\begin{proof}
 The construction of the hard distribution over inputs is as follows. The distribution will be over distances $d$ for a fixed set of $n$ points $\cX $. 
 
\begin{mdframed}
\textbf{Construction of the ground-truth metric $(\cX,d)$.}
\smallskip
\begin{enumerate}
\item Set $k = \frac{\log n}{100}$, and draw a uniformly random mapping $f:[n] \to [n/k]$ conditioned on $|f^{-1}(j)| = k$ for all $i \in [n/k]$. Define the $i$-th block $B_i = f^{-1}(i)$, and $B = \{B_1,\dots,B_{n/k}\}$.
\item Define the metric $d$ as follows:
\[     d(x,y) = \begin{cases}
    1 & \text{ if } (x,y) \in B_i \text{ for some } i \in [n/k] \\ 
    k & \text {otherwise }\\
\end{cases}\]

\end{enumerate}

\end{mdframed}
It is straightforward to verify that the construction of $d$ is a metric. We now describe how to generate the corrupted distances $\tilde{d}$ for a given draw of $d$. 

Note that the optimal MST first conencts together all points within the same block (each of the $\Theta(n)$ edges paying a cost of $1$ for each such edge), and then connects together the remaining $n/k$ blocks, each with a cost of $k$. Thus $\min_T w(T) = \Theta(n)$.

\begin{mdframed}
\textbf{Construction of the weak-oracle metric $\tilde{d}$}
\begin{enumerate}
\item Define the random graph $H = (\cX,\hat{E})$ as follow. We have $(x,y) \in E$, if and only if $x \in B_i, y \in B_j$ with $i\neq j$, and such that $(x,u) \in \cor$ and $(v,y) \in \cor$ for all $u \in B_j$ and $v \in B_i$. 
\item Let $M \subset \cX \times \cX$ be a maximum matching in the graph $H$.
\item Define the weak-oracle output $\td$ as follows: for every $(x,y) \in M$, where $x \in B_i,y \in B_j$, we set $\td(x,u) = 1$ and $\td(v,y) \in \cor$ for all $u \in B_j$ and $v \in B_i$. For all other pairs $(x,y)$, we set $\td(x,y)  = d(x,y)$.
\end{enumerate}
\end{mdframed}
First note, by the same argument in the proof of \Cref{fact:largeMatching}, we will have that $|M| > n/4$ with high probability. Note that even though the setting is slightly different, because $(x,y)$ can never be an edge if $(x,y)$ are in the same block $B_i$, there are still at least $n - n/k$ possible edges which can be adjacent to any individual point $x$, and each exists with probability at least $(1/3)^k > 1/\sqrt{n}$ as needed for the proof of \Cref{fact:largeMatching}. Call the event that $|M| > n/4$ $\cE_1$, and condition on it now.

We first prove the lower bound against algorithms that make no strong oracle queries. Now fix any draw of the observed distances $\td$, and condition on the identities of the matching $M \in \cX \times \cX$. By a simple averaging argument (Yao's Min-max principle), if there was a randomized algorithm correct with probability at least $1/\poly(n)$ against any given input, then there would be a deterministic algorithm correct against this input distribution with probability at least $1/\poly(n)$. So given such an algorithm, after fixing $\td$ we can fix the tree $T$ output by the algorithm. Like in the proof of \Cref{thm:mstLBMain}, we orient $T$ arbitrarily, let $\pi(x)$ be the parent of $x$ in $T$, and charge each vertex $x$ with the cost $d(x,\pi(x))$. 

Now note that for every pair $(x,y) \in M$, conditioned on the observations $\td$ and matching $M$, for this match pair we have $\td(x,u) = \td(y,u)$ for all $u \in \cX$. Thus, by the symmetry of the identities, for any fixed $B_i,B_j$ such that $x \in B_i, y \in B_j$ occurs with non-zero probability over the remaining randomness, we have that  $x \in B_i, y \in B_j$  occurs with the same probability as $x \in B_j, y \in B_i$. To see this formally, note that conditioned on any matching $M$, we can consturct a bijection between the remaining realizations of the randomness where $x \in B_i, y \in B_j$ and where $x \in B_j, y \in B_i$, simply by swapping the values of $f(x),f(y)$ -- this is possible because, after conditioning on any set of values $\{f(z)\}_{z \in \cX \setminus \{x,y\}}$, the marginals of $f(x)$ and $f(y)$ are identically distributed. Thus, for any fixed parent $\pi(x)$ of $x$, the probability that $\pi(x)$ is in the same block as $x$ is at most $1/2$. Thus the expected cost $d(x,\pi(x))$ for any matched point $x$ is at least $k/2$, since for any fixed block $B_i$ containing $\pi(x)$, we have $x \notin \pi(x)$ with probaiblity at least $1/2$. Since there are $\Omega(n)$ matched points, it follows that the expected cost of the algorithm is at least $\Omega(nk) = \Omega(n \log n)$, which completes the proof for the case of algorithms which do not query the strong oracle.

Finally, for any algorithm that makes at most $n/100$ strong (point) oracle queries, notice that there are still $n/20$ pairs of matched $(x,y)$ points such that neither were queried. For such pairs, the above claim still holds, namely that for any fixed $B_i,B_j$ such that $x \in B_i, y \in B_j$ occurs with non-zero probability,  both $(x \in B_i, y \in B_j)$ and $(x \in B_j, y \in B_i)$ occur with equal probability. Thus, for any parent $\pi(x)$, the point $x$ will be in a different block from $\pi(x)$ with probability at least $1/2$, even conditioned on the strong oracle observations, the matching $M$, and $\td$, and the rest of the proof proceeds as above. \myqed{\Cref{thm:mstLBGeneralMain}}
\end{proof}

\FloatBarrier
\section{Experiments}
\label{sec:experiment}
In this section, we experimentally validate the performances of our clustering algorithms. We compare our algorithms with benchmarks on two extremes: the ``weak baseline'', where the benchmark algorithm has access to only the $\WO$ queries, and the ``strong baseline'', where the benchmark algorithm has access to $\SO$ queries on the entire dataset. We demonstrate that:


\begin{enumerate}[leftmargin=10mm, label=(\roman*)]
\item The weak baseline algorithms with only $\WO$ access produce very poor-quality solutions; 
\item Our algorithm achieve costs that are competitive with the strong baseline that queries $\SO$ on the entire dataset, while only using $\SO$ queries on a very small fraction ($<1\%$) of the points.
\end{enumerate}



\paragraph{Datasets.} As discussed in \Cref{sec:intro}, our experiments use both synthetic data generated from the extensively-studied Stochastic Block Model (SBM) \cite{holland1983stochastic,dyer1989solution,decelle2011asymptotic,abbe2015exact,abbe2015community,hajek2016achieving,mossel2015consistency} and the embeddings generated from the MNIST dataset with t-SNE and SVD \cite{deng2012mnist,van2008visualizing}. We construct the SBM model with $k=7$ clusters: in the $i$-th cluster, we sample points from a Gaussian distribution $\mathcal{N}(\mu, I)$ with $\mu_{i}=10^5$ and $\mu_{j}=0$ for all $j\neq i$, and we use the $\ell_2$ metric. As points sampled from the Gaussian distribution are concentrated, ground truth clusters are well-separated and the cost of misclustering even a single point is large. For the MNIST dataset, we run t-SNE and SVD embeddings with $60k$ training data, and embed into $d=2$ dimensions for t-SNE and $d=50$ for SVD. 

In both scenarios, there are  clear ``ground truth'' clusters for each point. As such, there is a natural weak-oracle corruption policy: for a pair of points $(x_i, x_j)$, if $x_i$ and $x_j$ are in the same ground truth cluster, flip the distance to an arbitrary inter-cluster distance; otherwise, flip the distance to an arbitrary intra-cluster distance. For the synthetic dataset, this results in an SBM model. 

\paragraph{Algorithm Implementations.} We implement the weak and strong benchmarks with the farthest traversal algorithm \cite{gonzalez1985clustering} for the $k$-center task and the celebrated $k$-means++ algorithm for the k-means task \cite{ArthurV07}, and both are de facto choices in practice. To perform the Lloyds iteration for $k$-means, we reveal the embedding vectors on the points with the $\LQ$ query to the algorithm. We also use $k$-means++ as the post-processing algorithm of the sampled set $S$ in \Cref{alg:k-means}\footnote{For the weak benchmark employing $k$-means++, we reveal all embeddings for the Lloyd iterations, which only helps that baseline. However, for our algorithms, we only reveal embeddings of points queried in $S$.}. 
The basic version of the experiments are carried out Macbook Pro with M1 chip and 16GB RAM. An optimized version for larger-scale datasets was run on an virtual compute cluster with 360GB RAM.



\paragraph{Figures and tables.}  We vary the parameters for sampling in our algorithms and obtain the curves for the clustering cost vs.\ number of strong oracle queries for different values of $\delta$ (the weak oracle corruption probability).
For tables, in each setting of $\delta$ for different sampling parameters, we pick the run with the best query-cost trade-off by selecting the run that minimizes the value $\card{\SO} \cdot \cost^{10}$, where $\card{\SO}$ is the number of queries to the strong oracle. We do this in order to prevent selecting runs that make very few queries but have poor cost. 


\begin{table}[!h]
\caption{\label{tab:best-trade-off}The best query-cost trade-off point for the $k$-center and $k$-means algorithms on the SBM. `Competitive ratio' means the ratio between the costs of our algorithms and the storng benchmark. The left column indicates the percentages of $\SO$ queries used.}
\begin{center}
\begin{tabular}{ |*{8}{c|} }
    \hline
\multirow{ 2}{*}{} & $n$    & \multicolumn{3}{|c|}{$\%$ of data queried for $\SO$}
            & \multicolumn{3}{|c|}{Competitive ratio}
                                         \\
    \cline{3-8}
     & & $\delta = 0.1$ &  $\delta = 0.2$ &   $\delta = 0.3$  &   $\delta = 0.1$    &  $\delta = 0.2$  &   $\delta = 0.3$   \\
    \hline
\multirow{4}{*}{k-center} & 10k   &   6.51  &   7.42   &   7.42  &   0.828  &   0.707  &   0.880   \\
\cline{2-8}
& 20k   &   1.26  &   1.96  &   5.46  &   0.802  &   0.842  &   0.795   \\
\cline{2-8}
& 50k & 0.798 & 1.484 & 2.184 & 0.809 & 0.779 & 0.832 \\
\cline{2-8}
& 100k & 0.252 & 0.917 & 0.917 & 0.804 & 0.718 & 0.762 \\
\hline
\multirow{4}{*}{k-means} & 10k   &   5.55  &   3.51   &   13.19  &   1.089  &   1.053  &   1.175   \\
\cline{2-8}
& 20k   &   1.975  &   1.78  &   8.57  &   1.216  &   1.086  &   1.191   \\
\cline{2-8}
& 50k & 1.038 & 0.82 & 2.342 & 1.142 & 1.062 & 1.125 \\
\cline{2-8}
& 100k & 0.555 & 0.44 & 1.31 & 1.141 & 1.218 & 1.25\\
\hline
\end{tabular}
\end{center}
\end{table}

\paragraph{SBM Experiments}
\label{subsec:experiment-SBM}
We test with corruption rate of $\delta=0.1, 0.2$ and $0.3$
with the scales of $n=10k$, $20k$, $50k$, and $100k$. 
The cost (log scale) vs.\ strong oracle query curves and trade-off points for the $k$-center and $k$-means algorithms can be observed in \Cref{fig:sbm-k-center,fig:sbm-k-means} and \Cref{tab:best-trade-off}. In the plots of \Cref{fig:sbm-k-center}, weak baseline and strong baseline are farthest traversal with access to only $\WO$ and $\SO$ on entire dataset respectively. In comparison, for the plots of \Cref{fig:sbm-k-means}, the weak and strong baselines use $k$-means++ 
with zero $\SO$ queries and the entire set of $\SO$ queries, respectively.

As one would expect, in \Cref{fig:sbm-k-center}, the $k$-center cost decreases drastically at some thhreshold (from $>125k$ to $\sim 7$) since thereafter no point gets misclustered. Moreover, this threshold is quite small --- the algorithm converges as early as the point where it queries $\SO$ for only $\sim 0.5\%$ of the total points. In contrast, the drop of cost for $k$-means as in \Cref{fig:sbm-k-means} demonstrate a more ``smooth'' manner. Nevertheless, both \Cref{fig:sbm-k-center} and \Cref{fig:sbm-k-means} show that the costs of $k$-clustering algorithms drop sharply and approach the optimal cost with a very low percentage of $\SO$ queries. 

We then show in \Cref{tab:best-trade-off} the best query-cost trade-off points for the $k$-center and $k$-means algorithms. 
It can be observed that our algorithm consistently outperforms even the farthest traversal with $\SO$ queries on the entire dataset, while using queries only an extremely small fraction of the points. 
For the $k$-means experiments, our algorithm can provide a solution that is within a factor of $<1.25\times$ of strong benchmark with $\SO$ queries on $0.5\% \sim 1.31\%$ of the points in the dataset. 
We can also observe that trend from \Cref{tab:best-trade-off} that the percentage of $\SO$ queries decreases as $n$ becomes large, while the competitive ratio remains in the same range.
When $n$ is large (e.g., in the $100k$ case),  outperforming the benchmark takes $\SO$ queries on only $<1\%$ of the points.

\begin{figure}[!hbtp]
\centering
\begin{subfigure}{.45\textwidth}
  \centering
  \includegraphics[scale=0.33]{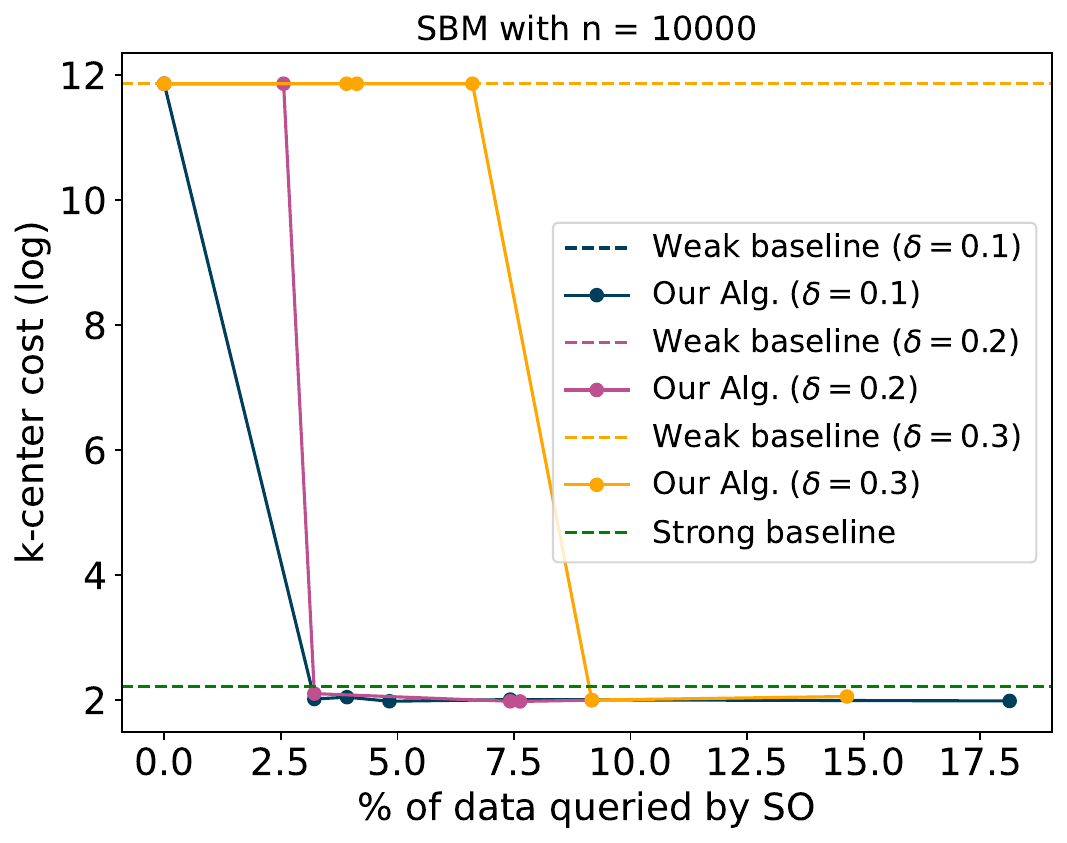}
  \label{fig:sbm10000-kcenter}
\end{subfigure}%
\begin{subfigure}{.45\textwidth}
  \centering
  \includegraphics[scale=0.33]{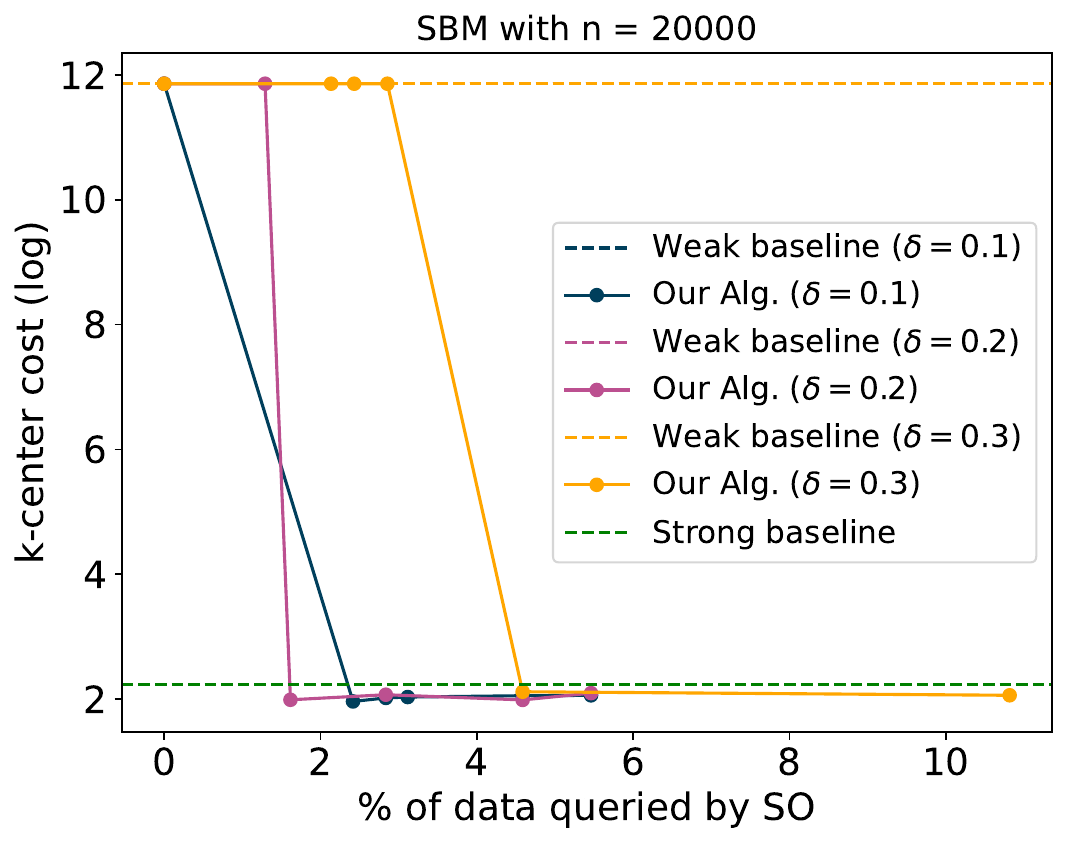}
  \label{fig:sbm20000-kcenter}
\end{subfigure}
\centering
\begin{subfigure}{.45\textwidth}
  \centering
  \includegraphics[scale=0.33]{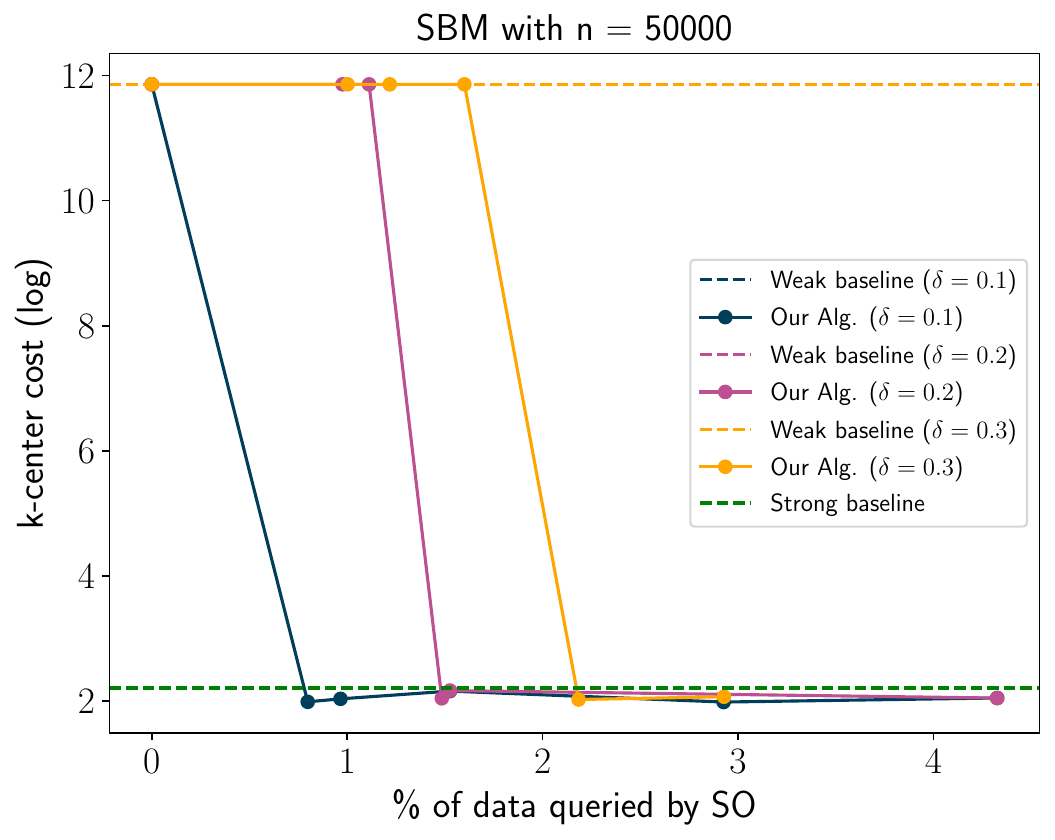}
  \label{fig:sbm50000-kcenter}
\end{subfigure}
\begin{subfigure}{.45\textwidth}
  \centering
  \includegraphics[scale=0.33]{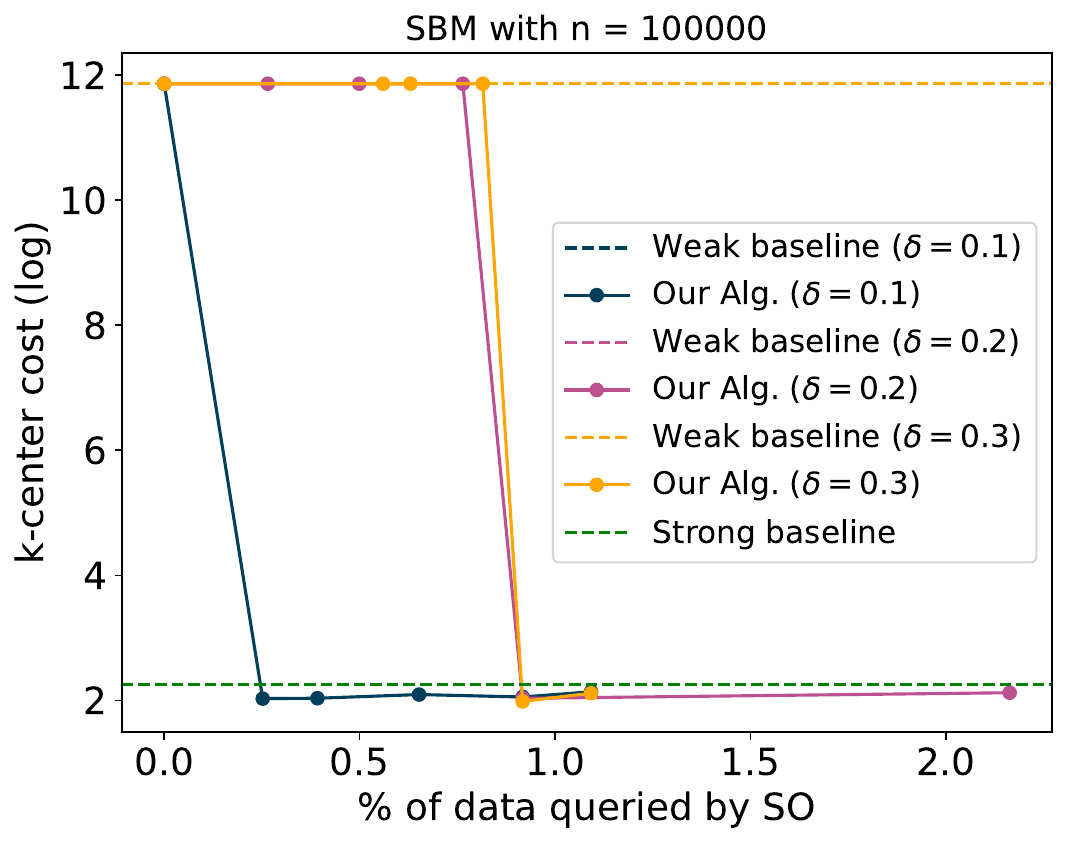}
  \label{fig:sbm100000-kcenter}
\end{subfigure}
\caption{Number of $\SO$ queries vs. clustering cost under the SBM model for \textbf{k-center} with different values of $\delta$ and $n$. Weak baseline - farthest first traversal with $\WO$ queries only and Strong baseline - farthest first traversal with $\SO$ queries on full dataset.
}
\label{fig:sbm-k-center}
\end{figure}

\paragraph{MNIST Experiments:}
We now discuss the results on MNIST with t-SNE and SVD embeddings. It is well-known that the MNIST t-SNE embedding with $d=2$ forms well-seperated clusters; however, the dichotomy between the distances of inter- and intra-cluster points are not as stark as the SBM model. Furthermore, separations between clusters is notably worse for the SVD embedding. Thus, the $t$-SNE and SVD datasets are ``less clustered'' and ``not clustered'' instances, respectively. 
\Cref{fig:mnist} snd \Cref{table:k_means_mnist} show the query-cost curve for MNIST with t-SNE and SVD embeddings. Compared to the SBM model, the curves for these embeddings decrease less rapidly, a consequence of the clusters not being as well-separated. 
Nonetheless, our $k$-means algorithm still outperforms the weak benchmark by a significant margin using a small fraction of $\SO$ queries.
For the t-SNE embeddings, as it is better-clustered, we observe a significant drop in cost after making less than $5\%$ of the $\SO$ queries. 
On the other hand, for the not-well-clustered SVD embedding, although there is only a factor of $\sim 1.2$ between the weak and strong benchmark costs, our algorithm still manages to achieve non-trivial improvements in cost beyond the weak benchmark with fewer than $5\%$ of the $\SO$ queries.



\begin{figure}
\centering
\begin{subfigure}{.45\textwidth}
  \centering
  \includegraphics[scale=0.34]{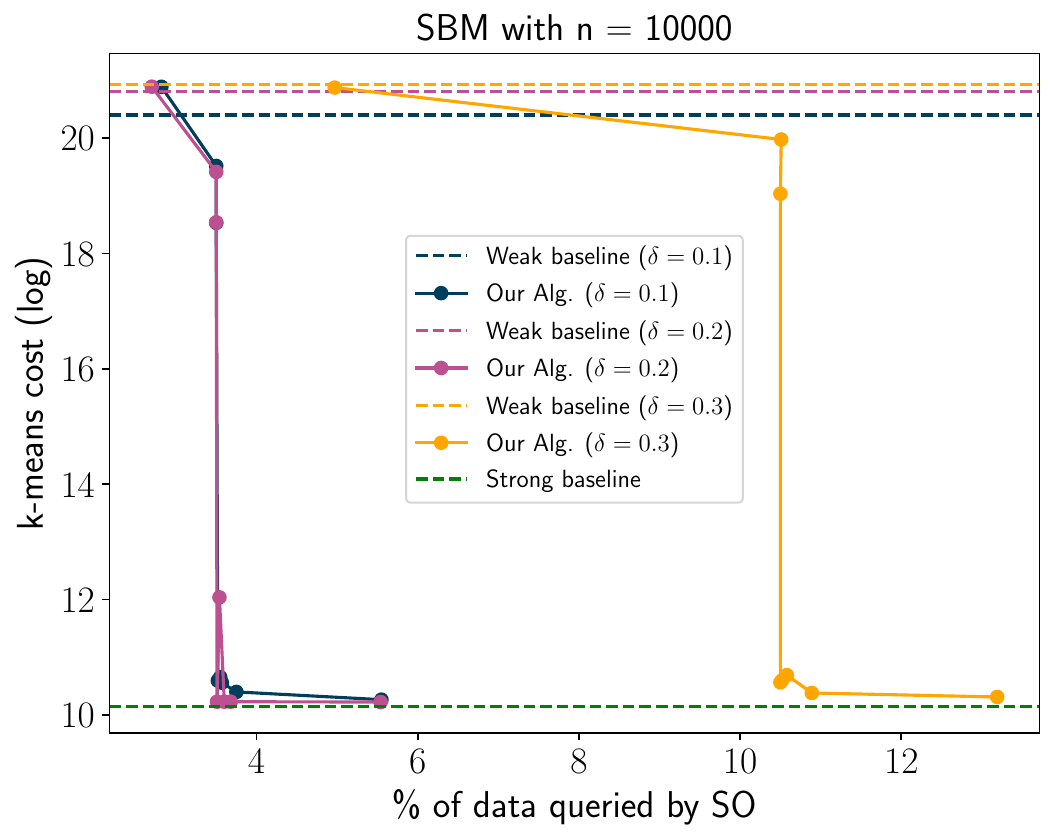}
  \label{fig:sbm10000-kmeans}
\end{subfigure}%
\begin{subfigure}{.45\textwidth}
  \centering
  \includegraphics[scale=0.34]{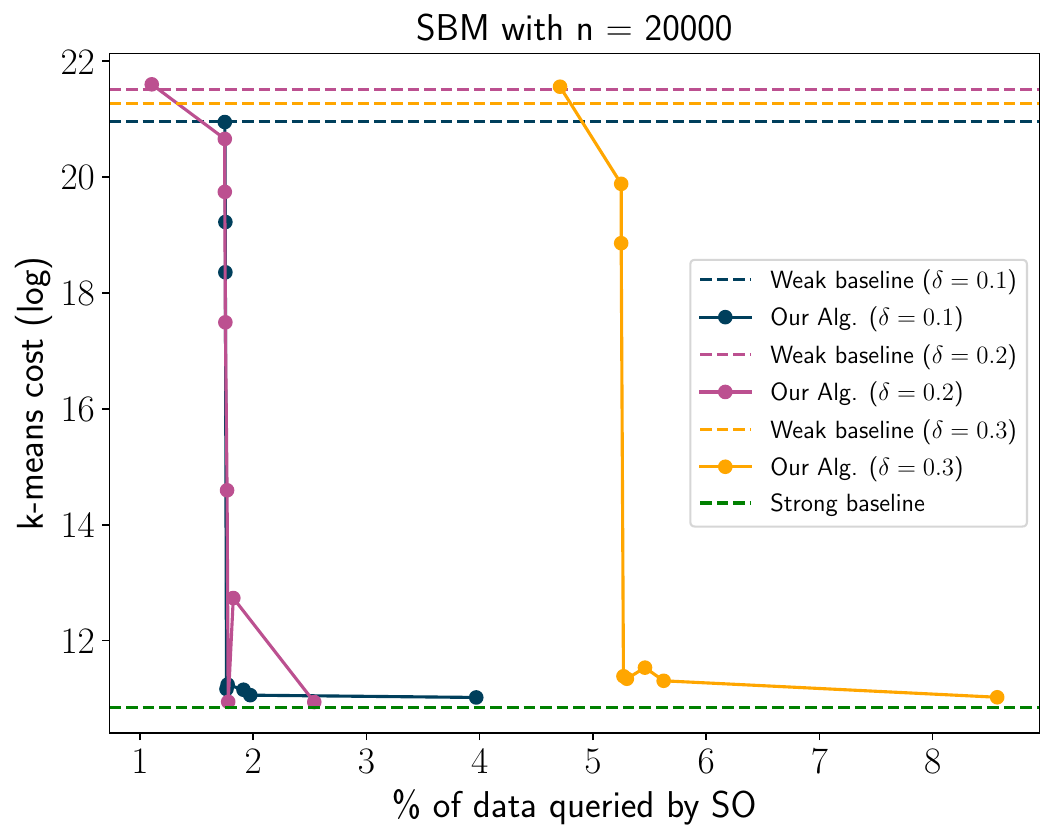}
  \label{fig:sbm20000-kmeans}
\end{subfigure}
\centering
\begin{subfigure}{.45\textwidth}
  \centering
  \includegraphics[scale=0.34]{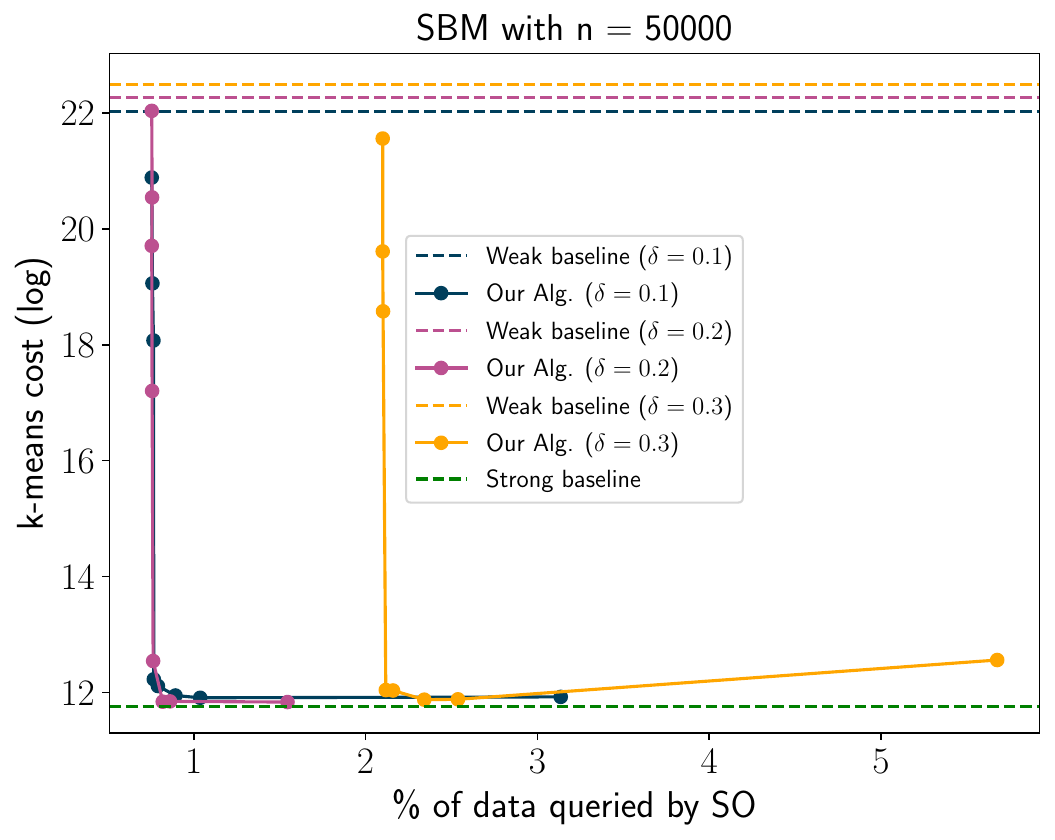}
  \label{fig:sbm50000-kmeans}
\end{subfigure}
\begin{subfigure}{.45\textwidth}
  \centering
  \includegraphics[scale=0.34]{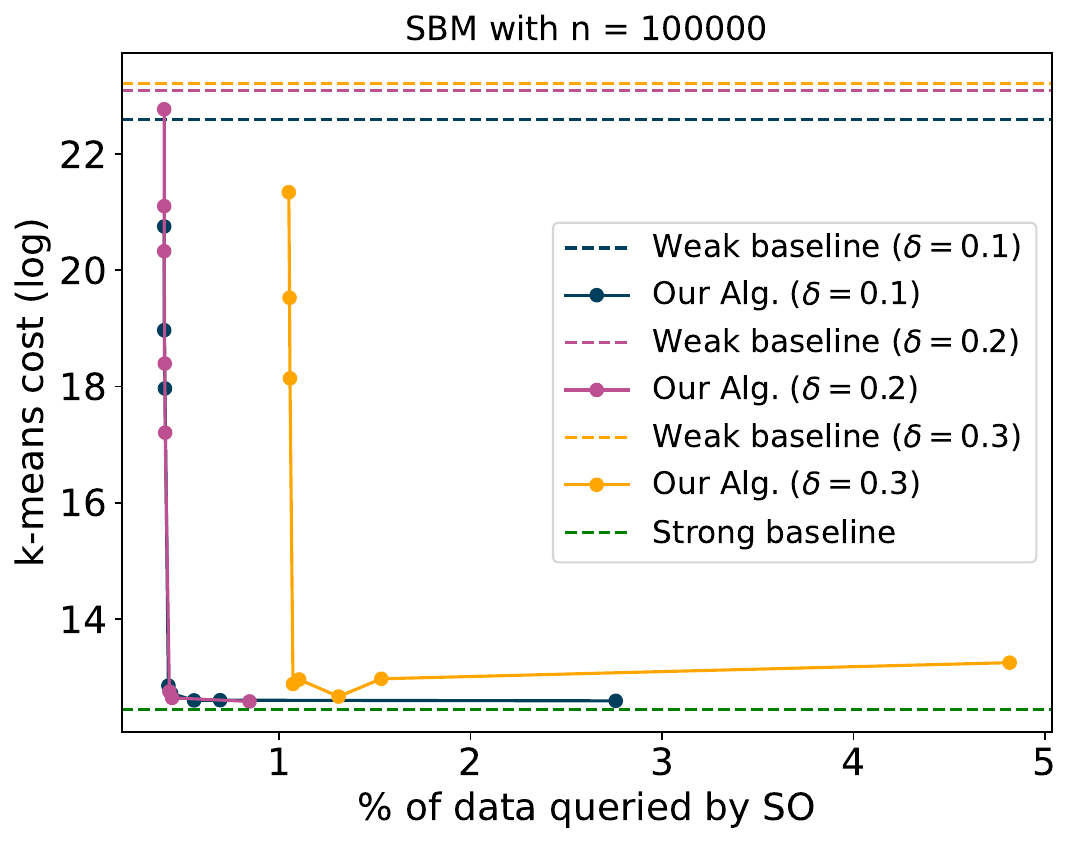}
  \label{fig:sbm100000-kmeans}
\end{subfigure}
\caption{Number of $\SO$ queries vs. clustering cost under the SBM model for \textbf{k-means} with different values of $\delta$ and $n$. Weak baseline - $k$-means++ algorithm with $\WO$ queries only and Strong baseline - $k$-means++ algorithm with $\SO$ queries on full dataset.
}
\label{fig:sbm-k-means}
\end{figure}

\begin{table}[!htb]
\caption{\label{tab:best-trade-off-k-means-MNIST}The best query-cost trade-off point for the $k$-means algorithm on the MNIST embeddings.}
\label{table:k_means_mnist}
\begin{center}
\begin{tabular}{ |*{7}{c|} }
    \hline
    & \multicolumn{3}{|c|}{$\%$ of data queried for $\SO$}
            & \multicolumn{3}{|c|}{Competitive ratio}
                                         \\
    \cline{2-7}
    &  $\delta = 0.1$ &  $\delta = 0.2$ &   $\delta = 0.3$  &   $\delta = 0.1$    &  $\delta = 0.2$  &   $\delta = 0.3$   \\
    \hline
t-SNE   &   4.58  &   4.57   &   6.62  &   1.169  &   1.286  &   1.367   \\
\hline
SVD   &   0.25  &   0.311  &   0.253  &   1.121  &   1.109  &   1.105   \\
\hline
\end{tabular}
\end{center}
\end{table}

\begin{figure}
\centering
\begin{subfigure}{.45\textwidth}
  \centering
  \includegraphics[scale=0.35]{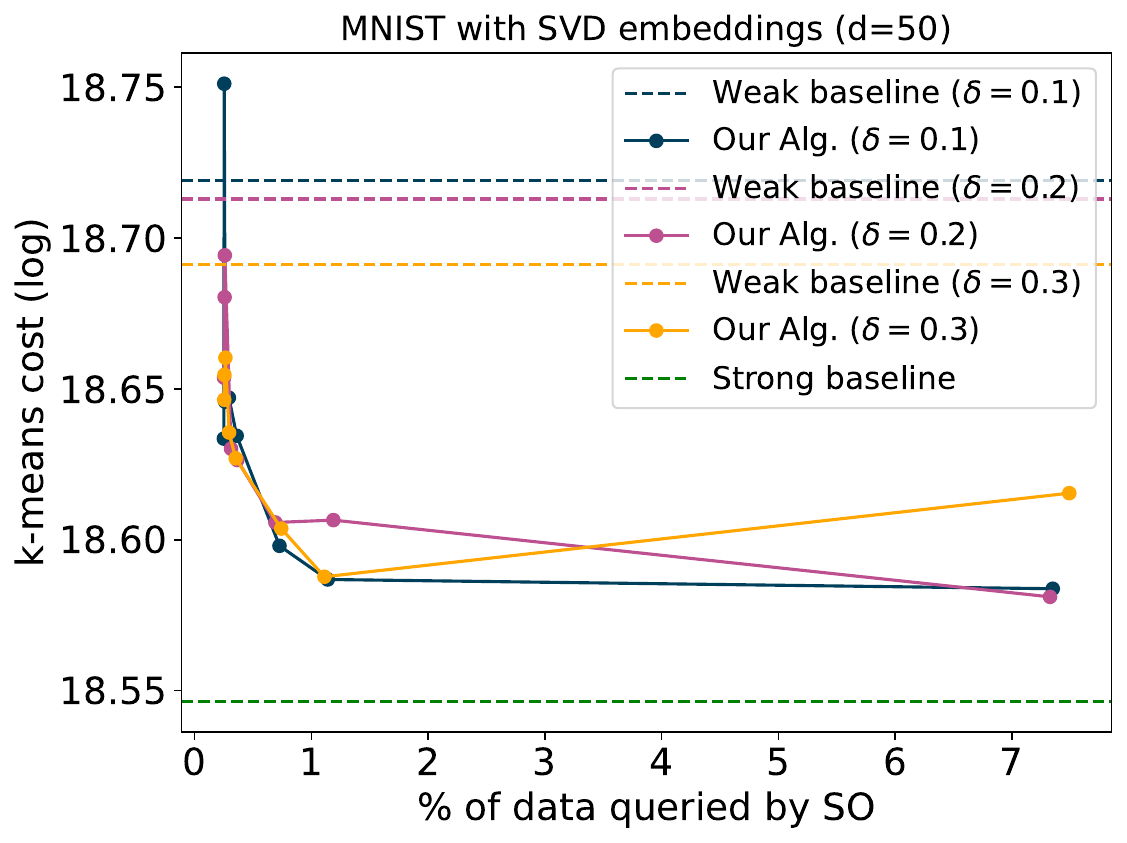}
  \label{fig:mnist-svd}
\end{subfigure}%
\begin{subfigure}{.45\textwidth}
  \centering
  \includegraphics[scale=0.35]{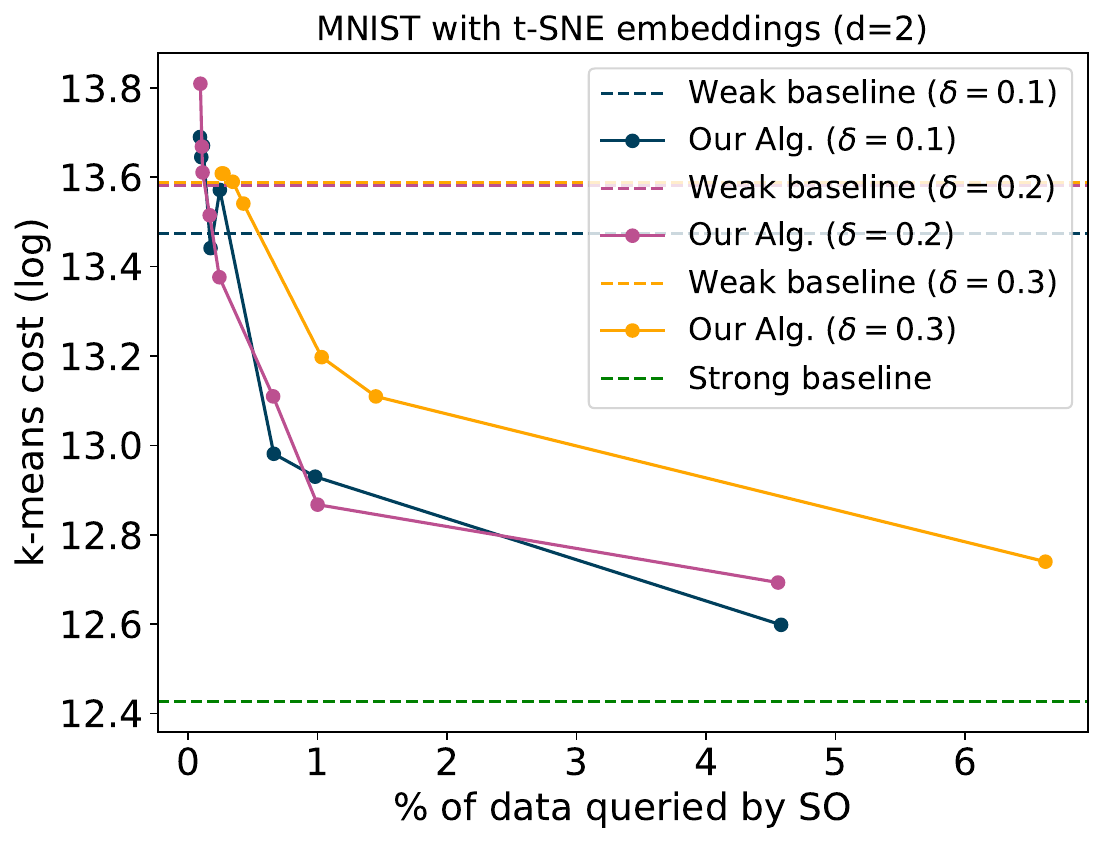}
  \label{fig:mnist-tsne}
\end{subfigure}
\caption{Number of $\SO$ queries vs. $k$-means clustering cost under the \textbf{MNIST dataset} with different values of $\delta$ and $n$. Weak baseline - $k$-means++ algorithm with $\WO$ queries only and Strong baseline - $k$-means++ algorithm with $\SO$ queries on full dataset.
}
\label{fig:mnist}
\end{figure}

\FloatBarrier



\bibliographystyle{plain}
\bibliography{references}
\newpage
\appendix

\section{Technical Preliminaries}
\label{sec:tech-prelim}

\subsection{Concentration inequalities}
We first state several standard concentration inequalities used in the analysis of this paper, beginning with the multiplicative form of the Chernoff bound.


\begin{proposition}[Multiplicative Chernoff bound]\label{prop:chernoff-mult}
	Let $X_1,\ldots,X_n$ be $n$ independent random variables with support in $[0,1]$. Define $X := \sum_{i=1}^{n} X_i$. Then, for every $\gamma >0$, there are 
	\begin{align*}
		& \Pr\paren{X > (1+\gamma)\cdot \expect{X}} \leq \exp\paren{-\frac{\gamma^{2} \expect{X}}{2+\gamma}}; \\
            & \Pr\paren{X < (1-\gamma)\cdot \expect{X}} \leq \exp\paren{-\frac{\gamma^{2} \expect{X}}{2}}.
	\end{align*}
\end{proposition}

\Cref{prop:chernoff-mult} assumes \emph{independent} random variables. It is known that the Chernoff bound is also applicable in the scenario of \emph{negatively correlated} random variables, defined as follows.

\begin{definition}[Negatively Correlated Random Variables]
\label{def:neg-cor-rvs}
Random variables $X_1,\ldots,X_n$ are said to be \emph{negatively correlated} if and only if
\begin{align*}
\expect{\prod_{i=1}^{n} X_{i}} \leq \prod_{i=1}^{n}\expect{X_{i}}.
\end{align*}
In particular, if $X_{i}$'s are independent, we have $\expect{\prod_{i=1}^{n} X_{i}} = \prod_{i=1}^{n}\expect{X_{i}}$.
\end{definition}

One can easily verify the following sufficient condition to verify the negative correlation between random variables.
\begin{fact}
\label{fct:neg-cor-rvs}
A sufficient condition for $X_{i}$ and $X_{j}$ to be negatively correlated is that conditioning on $X_{i}=1$, the probability for $X_{j}=1$ does not increase.
\end{fact}

And the concentration in \Cref{prop:chernoff-mult} applies to negatively correlated random variables.
\begin{proposition}[Generalized Chernoff]
\label{prop:chernoff-general}
Let $X_1,\ldots,X_n$ be $n$ negatively correlated random variables supported on $\{0,1\}$. Then, the concentration inequality in \Cref{prop:chernoff-mult} still holds.
\end{proposition}

\subsection{Aspect Ratio for Clustering and Guessing $\widetilde{OPT}$}
\label{seubsec:add-detail}

We assumed polynomially-bounded aspect ratio $\Delta$ in our work for the simplicity of presentation.
However, our results can be easily extended to arbitrarily large $\Delta$ by simply increasing the domain of possible guesses for $\widetilde{OPT}$. Specifically, note that only our clustering algorithms depend on the aspect ratio $\Delta$ (not our MST algorithm), and this dependency only appears in the universe of possible guesses for an approximation $\widetilde{OPT}$ of the optimal clustering cost (this approximation is needed for $k$-centers, $k$-means, and $k$-median). Specifically, consider the set of $O(\eps^{-1} \log \Delta)$ possible guesses $t= (1+\eps)^0, (1+\eps)^1,\dots,\Delta$ for the approximate cost $\widetilde{OPT}$, where we set $\eps = 1$ for $k$-means and $k$-medians, and allow for smaller $\eps$  for $k$-centers (as stated in Theorem \ref{thm:alg-k-center}). For each of the clustering problems, it will suffice to find \emph{any} level $t$ such that the guess at $t$ is deemed ``too small'' by the algorithm (described below), and such that running the algorithm at level $t(1+\eps)$ produces a valid solution. 

For $k$-centers, a level $t$ is deemed as ``too small'' of a guess if the number of centers produced by that level is larger than $k$ (see Algorithm \ref{alg:k-center}). For $k$-means and median, the level is ``too small'' if we sample too many (more than $O(k \log^2 m$) centers in Algorithm \ref{alg:k-means}. Thus, we can easily run our algorithm on guess of $\widetilde{OPT}$ which are chosen via binary search to find a guess $t$ which is ``too small''  and such that $t(1+\eps)$ is not too small (and therefore produces a valid solution). This results in a $O(\log \frac{\log \Delta}{\eps})$ factor in the strong oracle query complexity instead of a $O(\log \frac{\log n}{\eps})$ factor. Further problem-specific details are deferred to the proofs of the respective algorithms.

\section{Generalizing Our Clustering Algorithms to Larger $\delta$}
\label{appendix:deltaRemark}
We used the fixed corruption probability $\delta=\frac{1}{3}$ in our proofs of our clustering algorithms for clarity of presentation. However, we remark that our clustering algorithm works for arbitrary $\delta<\frac{1}{2}$ by scaling up the number of weak and strong oracle queries by a factor of $(\frac{1}{1/2-\delta})^2$. To see this, note that we only used the corruption probability in the median estimation, and our goal is show that in a fixed set $S$ of points, with high probability, for any $x\in \cX$, at least half of the distances between $x$ and $y\in S$ are not corrupted. In the Chernoff bound calculation, if the corruption probability is $\delta$, there are in expectation $\expect{X}=(1-\delta)\card{S}$ distances between $x \in \cX$ and $y\in S$ preserved for a fixed $S$. As such, the calculation of Chernoff bound becomes
\begin{align*}
\Pr\paren{X<\frac{\card{S}}{2}} &= \Pr\paren{X<\frac{1/2}{1-\delta}\cdot \expect{X}}\\
&= \Pr\paren{X<(1-\frac{1/2-\delta}{1-\delta})\cdot \expect{X}}\\
&\leq \Pr\paren{X<(1-(1/2-\delta))\cdot \expect{X}}\\
&\leq \exp\paren{-\frac{(1/2-\delta)^2\cdot \expect{X}}{3}}\\
&\leq \exp\paren{-\frac{(1/2-\delta)^2\cdot \card{S}}{6}},
\end{align*}
where the last two inequalities used the condition of $\delta<\frac{1}{2}$. As such, the condition of $\card{S}=O(\frac{1}{(1/2-\delta)^2}\cdot \log{n})$ suffices to keep the statement true with high probability.

\section{Ilustration of the clustering properties of MNIST tSNE and SVD embeddings}
\label{app:mnist-embedding-property}

We have mentioned in \Cref{sec:experiment} that there is a significant difference between the clustering costs for the MNIST embeddings obtained by tSNE and SVD methods. To give an intuitive justification of this statement, we include \Cref{fig:embeddings-comp} for the comparison of the embeddings. We note that these properties are well-known in the area, and we include them for completeness.

\begin{figure}[!h]
\centering
\begin{subfigure}{.5\textwidth}
  \centering
  \includegraphics[scale=0.23]{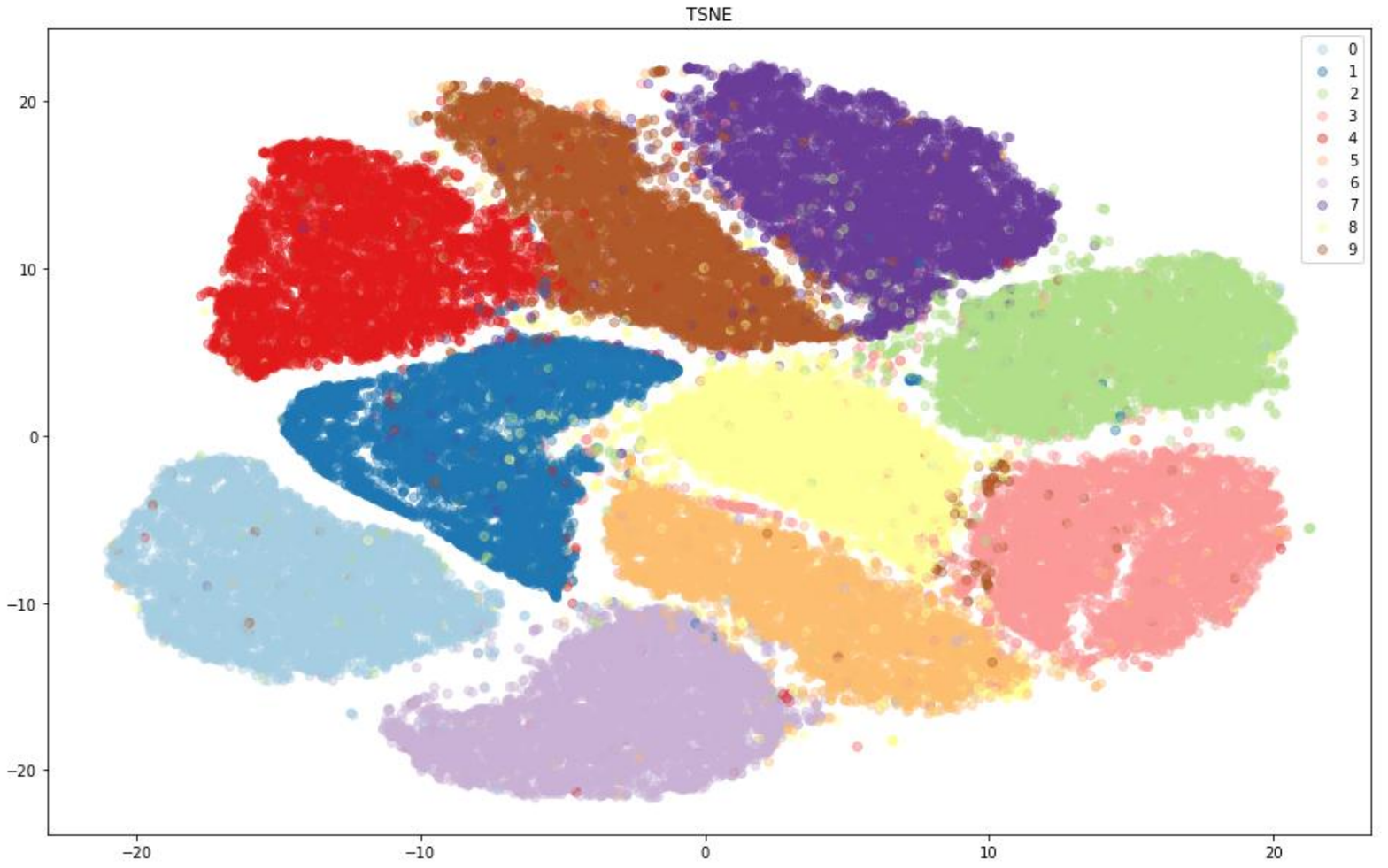}
  \caption{tSNE}
  \label{fig:tsne-embedding}
\end{subfigure}%
\begin{subfigure}{.5\textwidth}
  \centering
  \includegraphics[scale=0.23]{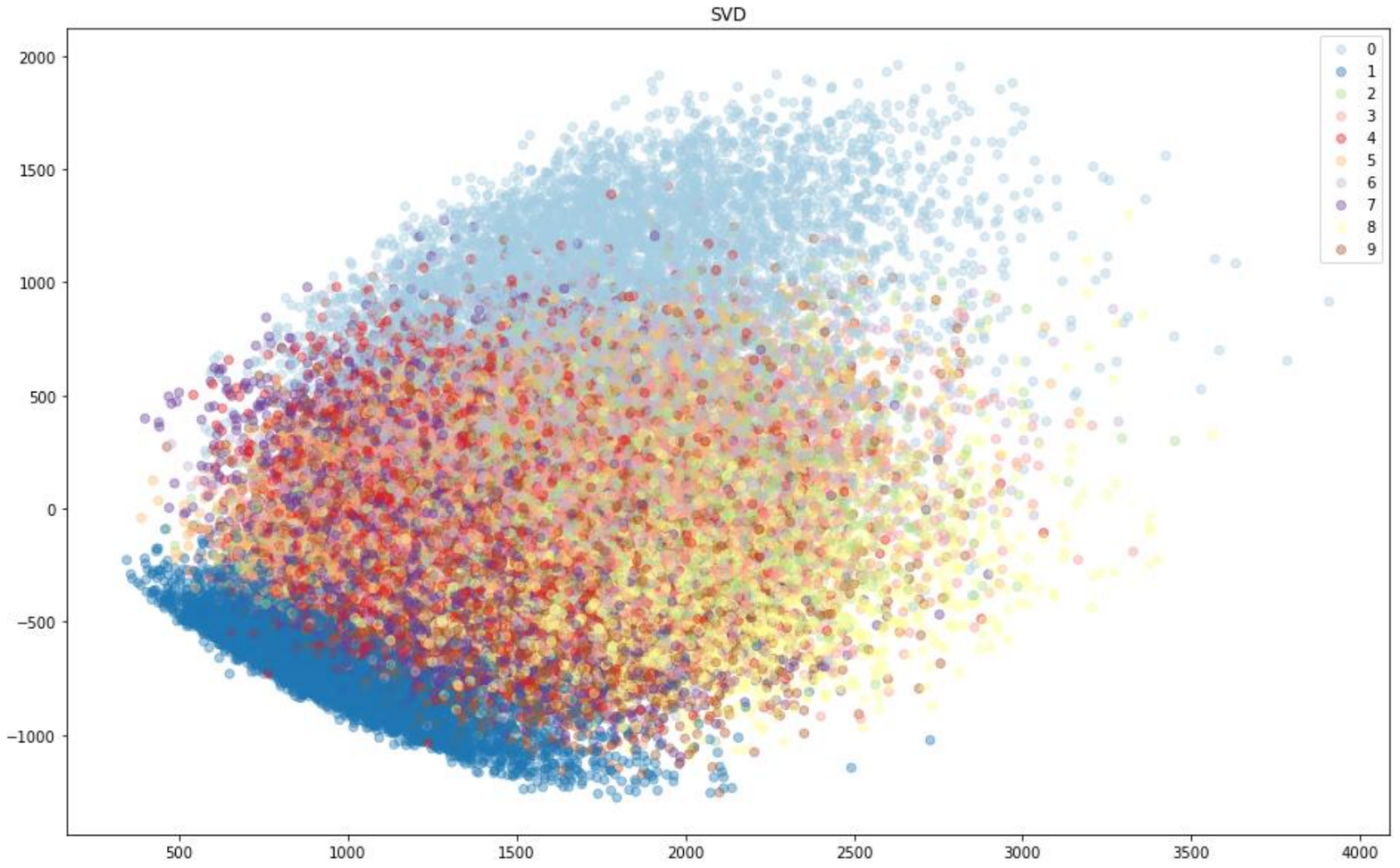}
  \caption{SVD}
  \label{fig:svd-embedding}
\end{subfigure}
\caption{The comparison of the MNIST 60k embeddings plot with 2 dimensions. Left: the embeddings generated by tSNE, which are well-clsutered; Right: the embeddings generated by SVD, and we pick the first 2 dimensions for plot. 
}
\label{fig:embeddings-comp}
\end{figure}

In \Cref{fig:tsne-embedding}, the embeddings of different classes are well-clustered in general, although the distances between the clsuters are not as large as the instanced generated by the Stochastic Block Model. In contrast, in \Cref{fig:svd-embedding}, the separation between classes is unclear, and the costs generated by a ``good'' and a random clustering are comparable.

\clearpage

\end{document}